\documentclass[journal]{IEEEtran}
\usepackage{comment}
\usepackage{todonotes}

\usepackage{subcaption} 
\usepackage{caption} 
\usepackage{sidecap}
\usepackage[utf8]{inputenc}
\usepackage[inline]{enumitem}

\DeclareCaptionFormat{myformat}{\fontsize{9}{9}\selectfont#1#2#3}
\sidecaptionvpos{figure}{c}
\usepackage{textcomp} 

\usepackage[ampersand]{easylist}
\usepackage{amsmath}
\usepackage{amsfonts}
\usepackage{amsthm}
\newtheoremstyle{mystyle}
{}
{}
{}
{}
{\itshape}
{.}
{ }
{}
\newtheoremstyle{mystyle2}
{}
{}
{\itshape}
{}
{\scshape}
{.}
{ }
{}

\newtheorem{mydef}{Definition}
\theoremstyle{mystyle}
\newtheorem{lemma}{Lemma}

\newtheorem{definition}{Definition}
\newtheorem{theorem}{Theorem}

\usepackage{soul}

\newcommand{\review}[1]{\begingroup\sethlcolor{yellow}\hl{#1}\endgroup}
\renewcommand{\review}[1]{#1}

\usepackage[T1]{fontenc}

\usepackage{graphicx}

\usepackage{listings}

\usepackage{caption}
\usepackage{subcaption}
\usepackage{mathtools}

\usepackage[]{units}
 
\usepackage{multirow} 
\usepackage{lipsum}
\usepackage{tikz}

\usetikzlibrary{positioning, math, shadings}
\usetikzlibrary{arrows,automata}
\usetikzlibrary{calc,decorations.pathreplacing}
\tikzset{router/.style={circle,draw,fill=gray!60,inner sep=0pt,minimum size=5pt}}
\usetikzlibrary{fadings}
\tikzset{desc/.style={font = \footnotesize}}

\newcommand{\makeNoCSectional}[2]{
\foreach \x in {0,...,3}{
	\node [router]  (\x1) at (2*\x,1) {};}
\foreach \x in {0,...,6}
{\pgfmathtruncatemacro{\label}{\x}
	\node [router]  (\x2) at (\x,0) {};}

\foreach \x in {0,...,5}{
	\pgfmathtruncatemacro{\y}{\x+1}
	\draw (\x2) -- (\y2);} 
\foreach \x in {0,...,2}{
	\pgfmathtruncatemacro{\xright}{\x+1}
	\draw (\x1) -- (\xright1);}
\foreach \x in {0,...,3}{
	\pgfmathtruncatemacro{\xbelow}{2*\x}
	\draw (\x1) -- (\xbelow2);}

\node[desc, left = 0.1 of 01, align=right, anchor = east] (layer1) {#1};
\node[desc, left = 0.1 of 02, align=right, anchor = east] (layer2) {#2};
}

\definecolor{col1}{RGB}{27,158,119}
\definecolor{col2}{RGB}{217,95,2}
\definecolor{col3}{RGB}{117,112,179}

\usetikzlibrary{shapes.symbols} 
\usetikzlibrary{shapes.geometric} 
\tikzset{
 hexagon/.style={signal,signal to=east and west}
}
\tikzset{
 octagon/.style={shape=regular polygon, regular polygon sides=8, draw, minimum width=.8in}
}

\usepackage{pgfplots}
\usepackage{booktabs}

\usetikzlibrary{decorations.pathreplacing, patterns}

\usepackage{algorithm} 
\usepackage{algorithmic}
\usepackage{diagbox} 
\DeclareMathAlphabet{\mymathbb}{U}{BOONDOX-ds}{m}{n} 
\newcommand{\zero}{\mymathbb{0}} 
\newcommand{\N}{\mathbb{N}}
\newcommand{\R}{\mathbb{R}}

\DeclareMathOperator*{\argmin}{arg\,min}
\newcommand{\norm}[1]{\left\lVert#1\right\rVert}

\newcommand{\ZXYZ}{Z\textsuperscript{+}(XY)Z\textsuperscript{-}}
\newcommand{\AlgorithmI}{Z\textsuperscript{+}(XY)Z\textsuperscript{-}}
\newcommand{\AlgorithmII}{ZXYZ}

\newcommand{\XYZ}{conventional XYZ}

\newcommand{\north}{\operatorname{north}}
\newcommand{\east}{\operatorname{east}}
\newcommand{\south}{\operatorname{south}}
\newcommand{\west}{\operatorname{west}}
\newcommand{\up}{\operatorname{up}}
\newcommand{\down}{\operatorname{down}}

\newcommand{\fnorth}{>>$\north$<<}
\newcommand{\feast}{>>$\east$<<}
\newcommand{\fsouth}{>>$\south$<<}
\newcommand{\fwest}{>>$\west$<<}
\newcommand{\fup}{>>$\up$<<}
\newcommand{\fdown}{>>$\down$<<}

\newcommand{\teast}{\emph{east}}

\newcommand{\twest}{\emph{west}}
\newcommand{\tup}{\emph{up}}
\newcommand{\tdown}{\emph{down}}

% ...
\newcommand{\clk}{\text{clk}}

\lstdefinestyle{routingScript}{
	emph={%
	if, else, then, route to
	},
    emphstyle={\bfseries},
	morekeywords={if, else, then, route, to, end},
	escapeinside={(*}{*)},
	numbers=left,
	stepnumber=2,
	numbersep=10pt,
	tabsize=2,
	showspaces=false,
	showstringspaces=false
}

\usepackage{textcomp}
\def\BibTeX{{\rm B\kern-.05em{\sc i\kern-.025em b}\kern-.08em
		T\kern-.1667em\lower.7ex\hbox{E}\kern-.125emX}}

\begin{document}
\bstctlcite{IEEEexample:BSTcontrol}

\title{NoCs in Heterogeneous 3D SoCs: Co-Design of Routing Strategies and Microarchitectures}


\author{
	\IEEEauthorblockN{{Jan Moritz Joseph}\IEEEauthorrefmark{1},{ Lennart Bamberg}\IEEEauthorrefmark{2}, {Dominik Ermel}\IEEEauthorrefmark{1}, {Behnam Razi Perjikolaei}\IEEEauthorrefmark{2}, {Anna Drewes}\IEEEauthorrefmark{1},  {Alberto Garc\'ia-Oritz}\IEEEauthorrefmark{2}, {Thilo Pionteck}\IEEEauthorrefmark{1}}
	\IEEEauthorblockA{\IEEEauthorrefmark{1}Otto-von-Guericke-Universit\"at Magdeburg\\
			Institut f\"ur Informations- und Kommunikationstechnik, 39106 Magdeburg, Germany\\
			Email: \{jan.joseph, dominik.ermel, anna.drewes, thilo.pionteck\}@ovgu.de}\\
		\IEEEauthorblockA{\IEEEauthorrefmark{2}University of Bremen\\
			Institute of Electrodynamics and Microelectronics, 28359 Bremen, Germany\\
			Email: \{agarcia, bamberg, raziperj\}@item.uni-bremen.de}}




\maketitle
\begin{abstract}
	Heterogeneous 3D System-on-Chips (3D SoCs) are the most promising design paradigm to combine sensing and computing within a single chip. A special characteristic of communication networks in heterogeneous 3D SoCs is the varying latency and throughput in each layer. As shown in this work, this variance drastically degrades the network performance. We contribute a co-design of routing algorithms and router microarchitecture that allows to overcome these performance limitations. We analyze the challenges of heterogeneity: Technology-aware models are proposed for communication and thereby identify layers in which packets are transmitted slower. The communication models are precise for latency and throughput under zero load. The technology model has an area error and a timing error of less than 7.4\% for various commercial technologies from 90 to \unit[28]{nm}. Second, we demonstrate how to overcome limitations of heterogeneity by proposing two novel routing algorithms called \AlgorithmI\ and  \AlgorithmII\ that enhance latency by up to 6.5$\times$ compared to conventional dimension order routing. Furthermore, we propose a high vertical-throughput router microarchitecture that is adjusted to the routing algorithms and that fully overcomes the limitations of slower layers. We achieve an increased throughput of 2 to 4$\times$ compared to a conventional router. Thereby, the dynamic power of routers is reduced by up to 41.1\% and we achieve improved flit latency of up to 2.26$\times$ at small total router area costs between 2.1\% and 10.4\% for realistic technologies and application scenarios.


\end{abstract}


\noindent \textbf{Accepted for publication in IEEE Access on Sept 10, 2019.}

\section{Introduction}

3D integration is one of the most promising paradigms to meet the perpetual demand for chips with higher performance, less power consumption and reduced area \cite{Dong.2009}. Therefore, many designs and architectures have been proposed: 3D-integrated DRAM subsystems, 3D-FPGAs \cite{Pavlidis.2010, Chaware.2015}, and even 3D-Vision Systems-on-Chip (3D VSoC) with stacked sensors \cite{Zarandy.2011}. Recently, Intel introduced "Lakefield", in which Foveros 3D technology is used to stack multicore processors, FPGAs and DRAM \cite{Intel.2019}. Other manufactures such as Xilinx are also targeting 3D integration \cite{Wu.2015}. Ultimately, stacking dies even tackles fundamental limits of computation by asymptotically reducing computation time from $t$ to $t^{0.75}$~\cite{Markov.2014}. All these works impressively demonstrate the advantages of 3D integration that are even exploited in commercial applications.


Despite the aforementioned incremental advancements, 3D integration enables one game-changing key innovation: It allows for heterogeneous integration, in which dies in disparate technologies, i.e.\ analog, mixed-signal, memory and logic are stacked. As stated in \cite{Lee.}, this is "the ultimate goal of 3D integration" because it allows to align the requirements of components with the technology characteristics of their die. This is advantageous for applications, in which components with different requirements are integrated to a single SoC: \cite{Garrou.2009} introduces an architecture for Internet of Things (IoT) stacking wireless sensors, RF communication, data processing and energy scavenging. In high-performance processors, interleaving of dedicated dies with either memory or processing increases performance \cite{Yu.2013}, with exemplary designs \cite{Sun.2013, Kikuchi.2011, Abe.2008, Kim.2015}. Especially, vision applications can profit of heterogeneity: 3D VSoCs \cite{Zarandy.2011} combining image sensing, mixed-signal conversion and digital image processing. Thus, VSoCs demand heterogeneous integration, intrinsically. In recent research, a SoC is proposed for self-driving cars that realizes up to 10,000 frame per second 
\cite{Koyanagi.2019}. Going further, the mixed-signal layers can implement analog accelerators, for instance to calculate a cellular neural network \cite{Kim.2009}. Such accelerators have been implemented in \unit[180]{nm}~\cite{Jia.2017} and \unit[130]{nm} \cite{Ghaderi.2015} CMOS technology. To summarize, heterogeneous integration enables to build novel and more efficient systems in previously challenging application areas.

To unleash the full potential of 3D integration, the used interconnection architectures must offer very good PPA (power, performance, area). In general, there are two approaches to distribute interconnects in the third dimension: First, the components of the interconnect architecture can be distributed in 3D. For instance, \cite{Dubois.2013} enables packet transmission in vertically partially connected 3D NoCs using elevator-first routing. \cite{Kim.2015} presents a NoC that connects cores for a neural network using TSVs. Also, works on inductive coupling have been made \cite{Miura.2013}. Second, the components of the interconnect itself can be split-up over layers and be distributed. For instance, MIRA \cite{Park.2008} is such a 3D stacked router that achieves up to 51\% latency improvement for synthetic workloads. While all these works are well-suited for interconnects on homogeneous 3D integration, they do not explicitly account for varying integration properties within the interconnect from heterogeneous 3D integration.

Since the integration properties of any interconnect will differ if it spans multiple heterogeneous layers of a chip, heterogeneous 3D interconnects must account for this property. There are two main integration issues as illustratively shown in Fig.~\ref{fig:challenges}: First, \emph{the components of an interconnection architecture are not purely synchronous}, since logic in digital nodes can be clocked faster than in mixed-signal nodes; the clock deviation can be by a magnitude and not only a small deviation. Second, \emph{the feature size of (identical) components differs with technology}. Traditional router architectures cannot be applied, because these cannot cope with different clock speeds or yield unbearable costs in mixed-signal layers. This will be discussed in Sec.~\ref{sec:related} in detail. \review{Because of the aforementioned two arguments, novel models, architectures and concepts are required:} For instance,  \cite{Bamberg.2018} proposes TSV power models that account for heterogeneity and low-power coding with less than 1\% error compared to bit-level accurate simulations. In another example, in ref. \cite{Joseph.2017} input buffer distributions among layers are evaluated with area saving between 8.3\% and 28\% and power savings between 5.4\% and 15\%. More significant improvements are required, which also improve performance. In particular, latency and throughput in heterogeneous 3D interconnects vary per layer due to different clocks and router count. These severe effects of heterogeneity were, previously to this work, unconsidered for heterogeneous 3D interconnects. 

\begin{figure}
	\centering
%
%
\begin{tikzpicture}

\node[desc, align = right, anchor = east] (dn) at (0.5,1) {digital nodes};
\node[desc, align = right, anchor = east] (msn) at (0.5,0) {mixed-signal nodes};
\node[desc, align = right, anchor = east] (issues) at (0.5,-1) {\textbf{integration issues:}\\ \textbf{components are\dots}};
\node[desc, align = center, anchor = center] (clock) at (2,2) {\textbf{clock speed}};
\node[desc, align = center, anchor = center] (fs) at (4.375,2) {\textbf{feature size}};
\node[desc, align = center, anchor = center] (issues1) at (2,-1) {\textbf{not purely}\\ \textbf{synchronous}};
\node[desc, align = center, anchor = center] (issues2) at (4.395,-1) {\textbf{varying in size}\\ \textbf{and number}};

\draw[thick] (1,0.75) -- ++ (0.25, 0) -- ++ (0, 0.5) -- ++ (0.25, 0) -- ++ (0, -0.5) -- ++ (0.25, 0) -- ++ (0, 0.5) -- ++ (0.25, 0) -- ++ (0, -0.5) -- ++ (0.25, 0) -- ++ (0, 0.5) -- ++ (0.25, 0) -- ++ (0, -0.5) -- ++ (0.25, 0) -- ++ (0, 0.5) -- ++ (0.25, 0) -- ++ (0, -0.5);
\draw[thick] (1,-.25) -- ++ (0.5, 0) -- ++ (0, 0.5) -- ++ (0.5, 0) -- ++ (0, -0.5) -- ++ (0.5, 0) -- ++ (0, 0.5) -- ++ (0.5, 0) -- ++ (0, -0.5);

\draw[thick, col1] (4.5-0.25,1.2+0.125) -- (4.5,1.2+0.125);
\draw[thick, col1] (4.5-0.25,0.7+0.125) -- (4.5,0.7+0.125);

\draw[thick, col1] (4.5-0.5+0.125,1.2+0.125) -- (4.5-0.5+0.125,0.7+0.125);
\draw[thick, col1] (4.5+0.125,1.2+0.125) -- (4.5+0.125,0.7+0.125);

\foreach \x in {4.5-0.5,4.5}{
	\foreach \y in {0.7+0.5,0.7}{
	\draw[thick, col1, fill = col1!10] (\x, \y) rectangle ++(0.25,0.25);
}
}
\draw[thick, col1] (4.5-0.75,0) -- (4.5,0);
\foreach \x in {4.5-0.75, 4.5}{
		\draw[thick, col1, fill = col1!10] (\x, -.25) rectangle ++(0.5,0.5);
}

\draw (.75,2.3) -- (.75,-1.3);
\draw (3.4,2.3) -- (3.4,-1.3);

\draw(-2.2, 1.7) -- +(7.5,0);
\draw(-2.2, 0.5) -- +(7.5,0);
\draw(-2.2, -0.5) -- +(7.5,0);

\end{tikzpicture}
	\caption{Challenges for heterogeneous interconnection architectures.}
	\label{fig:challenges}
\end{figure}

The aforementioned influence of heterogeneity on interconnects requires a novel approach that \emph{simultaneously} considers routing strategies and architectures; in a separate design, the full potential of the interconnect cannot be unleashed since either throughput or latency are limited. Therefore, this paper provides the following specific contributions:
\begin{enumerate}[topsep=0pt, label=\textit{Contribution \arabic*:}, leftmargin=\widthof{\textit{Contribution 1}}+\parindent]
	\item We introduce models for network throughput and latency; and we thereby show that heterogeneity drastically degrades network performance.
	\item We contribute two new principles for routing and two concrete routing algorithms reducing latency. The algorithms exploit the variation in communication speed between layers.
	\item \review{We propose a novel co-designed router and routing strategies} that tackles throughput, which is limited by the slowest router in a packet's path. The router architecture and the used simulation tools to generate results are published open-source.
\end{enumerate}
By these contributions, we tackle throughput and latency limitations in heterogeneous interconnects by an integrated approach. The source code of the simulation tool and the router architecture are available at \cite{Joseph.Github}.

The work is structured as follows: 
We discuss limitations of related approaches for heterogeneous 3D SoCs (Sec.~\ref{sec:related}).
To quantify the effects of heterogeneity, we propose a model for technology (Sec.~\ref{sec:model}) and communication (Sec.~\ref{sec:model:comm}). We thereby show a drastically negative impact on the network performance due to slower packet provision (Sec.~\ref{sec:potentials}). Based on these findings, we contribute two novel routing algorithms to overcome this issue by improving the latency (Sec.~\ref{sec:routing}). Further, we propose a router architecture adopted to these routing algorithms, which fully nullifies the negative influence of heterogeneity and improves network throughput (Sec.~\ref{sec:architectures}). Finally, we present the accuracy of our models and that latency, throughput and dynamic power are improved at minor hardware costs (Sec.~\ref{sec:results}). \review{In this section, we also present a comprehensive case study for a realistic system that demonstrates positive effects of our approach under practical conditions including congestion.} We thereby discuss \review{in Sec.}~\ref{sec:discussion} all relevant aspects of routing in NoCs for heterogeneous 3D SoCs and contribute that, solely, a co-design of algorithms and architectures allows for efficient heterogeneous 3D interconnects.

\section{Related work}\label{sec:related}

As already stated in the introduction, we discuss here why existing 3D interconnects are not considering heterogeneity sufficiently. Therefore, we focus on three individual topics that are also covered by this work: First, we highlight the \review{differentiating} aspects of our models for communication in heterogeneous 3D interconnects in comparison to other models. Second, we turn the spotlight to routing and we discuss related approaches in 3D systems. Third, we consider existing architectures for 3D interconnects. 
Finally, we combine the approaches and argue, why these existing methods and architectures are not well-applicable to build heterogeneous 3D interconnects. 

Modeling properties of both technology and communication in interconnects is a well-established research topic with a wide range of works. There are many works on 3D NoCs modeling performance (e.g. \cite{Kiasari.2013}) or power and area (e.g.~\cite{Kahng.2012}). The majority of the performance models can be applied \review{only} under zero load because non-dynamic behavior is easier to model. For instance, in \cite{Nikitin.2009}, a performance model is proposed with focus on Quality of Service (QoS). It assumes constant service time and purely synchronous routers. This cannot be applied to heterogeneous 3D interconnects, as those are not purely synchronous. \cite{Ogras.2010} models average latency, throughput and network characteristics without QoS guarantees. Again, this model cannot be applied for heterogeneity, because the model assumes \review{one globally synchronized clock}. In a similar approach, \cite{Arjomand.2009} models performance and power of NoCs with wormhole routers. Again, only homogeneous router architectures and technologies are covered. In a more sophisticated approach, the dynamic properties, namely load and congestion, are covered by some works, as well, e.g.\ \cite{Foroutan.2010, Kiasari.2013}.  A common approach is the use of queueing models \cite{Kiasari.2013}, in which the dynamic behavior is summarized by network statistics. Although there has been considerable effort to analytically model the behavior of interconnection networks, there is an urgent need for models for heterogeneous 3D interconnects as their properties, especially differences in clock speeds within a network, have not been considered sufficiently so far.

Routing for 3D interconnects is, just as models, a very common field, as well. The most traditional approach is the extension of strategies from 2D by one dimension. For instance, dimension-ordered routing can be directly used in 3D, which \review{has} already been done over a decade ago \cite{Kim.2007}. Since then, many improvements have been proposed: DyXYZ \cite{Ebrahimi.2013} is a \review{fully adaptive} routing considering congestion in 3D. Elevator-first routing \cite{Dubois.2013} enables packet transmission in vertically partially connected 3D NoCs. LA-XYZ \cite{Ahmed.2012} uses look-ahead strategies to improve latency and throughput by approximately 45\% and while reducing power by 15.9\%. Furthermore, fault-tolerance can be implemented, as well \cite{Ahmed.2016}. Despite the large number of papers in this area, none of these works target heterogeneous 3D interconnects, in which the transmission speed is not only impeded by congestion but also and more fundamentally by the used technology nodes.

\review{Architectures for 3D interconnects have also been researched for many years. These solutions manly target performance increases:} \review{Ref. }\cite{Kim.2007} was one of the first routers for 3D systems. Extending standard architectures, \cite{Krishna.2008} proposes express virtual channels (EVC) which combine conventional full-swing, short-range wires and low-swing, multi-drop wires, achieving 25\% latency reductions. This ultimately lead to the single-cycle NoC router SMART \cite{Chen.2013} with 60\% latency savings. Rather popular is MIRA \cite{Park.2008}, which was the first 3D-stacked router. All router components, except central arbitration, are sliced in logical-equal parts and are distributed. Thereby, up to 51\% latency improvement for synthetic workloads are achieved. \review{For heterogeneous 3D integration, architectures for not-purely synchronous communication are highly relevant. There exist only a limited number of works: Ref.} \cite{Vivet.2017} proposes a router architectures which is limited by the slower clock frequency for packets traveling along the asynchronous path. However, enabling asynchronous communication between routers is not a common topic due to its large overhead and limited practical relevance in homogeneous 3D systems. The limitations of non-purely synchronous communication between routers as intrinsically found in heterogeneous 3D interconnects is a key for their integration and, therefore, \review{is one of the key contributions of this publication}.

To summarize our discussion of the related work, none of the aforementioned works target the special requirements of heterogeneous 3D integration. \review{In terms of \textit{models}, there exist no well-known works on latency and throughput for heterogeneous 3D interconnects. The majority of \textit{routing} \textit{algorithms} for 3D interconnects do not consider performance differences between routers due to varying technology nodes; yet, this effect is significant as we will show in this work. The works on \textit{architectures} for 3D interconnects assume synchronous routers; yet, routers in heterogeneous 3D systems are not clocked purely synchronous, as this paper also will show. However, works on asynchronous routers do not target to increase the vertical link bandwidth to bridge the throughput gap posed by heterogeneity. Rather, they decrease the bandwidth to increase yield from TSV manufacturing, e.g. by serialization} \cite{Darve.2011}. \review{This is orthogonal to our targets and therefore, these approach cannot be used. Also, distributed architectures such as MIRA cannot be applied to heterogeneous 3D SoCs: First, processing elements would need to be equally distributed among all layers, but are actually located in that layer best suited for their technological requirements. Second, router delay is limited by the slowest layer and router area is dominated by the most expensive layer. To the best of our knowledge, there are no related works which consider the relationship between routing algorithms and architectures but this topic is very relevant in heterogeneous 3D SoCs due to latency and throughput limitations. Therefore, we see an urgent need to tackle these issues in one integrated design approach: Efficient heterogeneous 3D interconnects are only possible by means of a simultaneous design of routing algorithms and architectures, as demonstrated by this paper.}







\section{Modeling technology heterogeneity in 3D interconnects}\label{sec:model}

Heterogeneity influences every metric of the interconnect and we model the influence on area and timing, which are most relevant. The model accounts for any type of commercial technology and any feature size. We do not model power due to the diverse influence parameters; e.g.\ data transmitted vastly influence power consumption of links, which is hard to model a priori without simulations \cite{GarciaOrtiz.2017}.

We start this section by our technological assumptions for the models. We model an interconnect with NoC routers vertically connected via vertical interconnects. These interconnects can either monolithic inter-tier vias (MIVs) or trans-silicon vias (TSVs) due to their high interconnection density. 
We use the following model assumptions:
\begin{enumerate}
	\item The delay of a vertical interconnect is negligibly small, in comparison to horizontal and logic delay. The reason lies therein that vertical interconnects have a constant length of \unit[50]{\textmu m} due to substrate thickness \cite{.2012l}.
	\item We neglect modeling area of vertical interconnects because MIVs nearly have no overhead and TSVs have a constant one. 
	\item Vertically connected routers must not be located at the same physical 2D position (in their layer). Vertical links and routers can be horizontally connected via  redistribution. This variability is limited by the link delay. We model this by conversion of router locations to router addresses. 
	\item We show advantages of our approach in terms of power using simulations only. We do not model the different power properties of horizontal and vertical interconnects as this is a complex topic on its own. For models we kindly refer to \cite{Bamberg.2018, Joseph.2019b}. 
	\item We model synchronous routers within layers and not purely synchronous routers between layers, following a GALS approach (globally asynchronous, locally synchronous). This is reasoned as follows: Heterogeneous 3D interconnects will be in non-purely synchronous settings, since components in disparate technologies are potentially clocked at varying speeds and the slowest, synchronous clock wastes performance. Routers within layers, however, are in the same technology and therefore are clocked synchronous. 
\end{enumerate}
To summarize, the chosen assumptions are the most common integration principles for 3D interconnects and therefore are a reasonable choice.

Before introducing our models, we explain our notations and definitions: We consider a chip with $\ell$ layers and their index set $[\ell] = \{1, \dots, \ell\}$. We assume $n$-$m$-mesh topologies of NoCs per layer. The \review{feature size} of the technology nodes of layers, measured in [nm], is given by $\tau: [\ell] \rightarrow \mathbb{N}$. We call a chip layer with index $\iota$ \textit{>>more advanced node<<} than a layer with index $\xi$ if $\tau(\iota) < \tau(\xi)$ (for easy notation). We define:
\begin{mydef}[Relative technology scaling factor]
	Let $\xi$ and $\iota$ be the indexes of layers with technologies  $\tau(\xi)$ and $\tau(\iota)$ and with $\tau(\xi) > \tau(\iota)$. The \emph{relative technology scaling factor} $\Xi$ is:
	\begin{equation} \Xi(\xi, \iota) := \frac{\tau(\xi)}{\tau(\iota)}\end{equation}
\end{mydef}


\subsection{Area model}

The area, which the communication infrastructure in a layer requires, is influenced by two major factors: The size of an individual router and the number of routers. The effect of both factors is encapsulated into an abstract model. It covers the influence of technology nodes, constraints of synthesis tools and router architectures.

\subsubsection{Area of routers}\label{sec:technolgyScalingFactor}
Routers in layers in mixed-signal nodes are disproportionately expensive: While routers still will consist of conventional digital circuits, the technology node, e.g. mixed-signal technology, impacts on the size of routing computation, crossbars and buffers, affecting bot combinational and sequential logic. The overall area consists of logic, for which it is commonly known that it reduces its size (ideally) quadratically for more advanced nodes, and the remainder (e.g.\ power supply) that does not scale approximately and therefore remains constant for different nodes. This is shown illustratively in Fig.~\ref{fig:technologzareascling}. These considerations yield an area model of the form $\hat{\alpha} + \alpha \varsigma^2$, in which $\hat{\alpha}$ is the constant part \review{(i.e.~the non-scaling part)}, $\alpha$ is an non-ideality factor (i.e.~the deviation of the ideally quadratically scaling parts), and $\varsigma$ is the feature size. By this model we define the area scaling factor as the difference between baseline technology, i.e.\ the largest node, and any target technology:
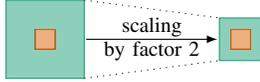
\begin{figure}
	\centering
	\begin{tikzpicture}[scale = .26]
			\draw[draw=col1, fill = col1!50] (-2,-2) rectangle (2,2);
			\draw[fill = white](-.5,-.5) rectangle (.5,.5);
			\fill[draw=col2, fill = col2!50] (-.5,-.5) rectangle (.5,.5);
			
			\draw[draw=col1, fill = col1!50] (10-1.09,-1.09) rectangle (11.09,1.09);
			\draw[fill = white](9.5,-.5) rectangle (10.5,.5);
			\fill[draw=col2, fill = col2!50] (9.5,-.5) rectangle (10.5,.5);
			
			\draw[-latex] (2.1,0) -- (10-1.09-.1,0) node[desc, midway, align = center]{scaling\\ by factor 2};
			\draw[dotted]  (2,2) -- (10-1.09,1.09) ;
			\draw[dotted]  (2,-2) -- (10-1.09,-1.09) ;
			%
			%

	\end{tikzpicture}
	\caption{Area scaling has reducing parts (green) and constant parts (orange).}
	\label{fig:technologzareascling}
\end{figure}
\begin{mydef}[Area scaling factor]
	Let $\xi$ and $\iota$ be the indices of two chip layers with technologies  $\tau(\xi)$ and $\tau(\iota)$ and with technology difference $\Xi (\xi, \iota)$. The \textit{area scaling factor } $s_f: (\mathbb{R}) \rightarrow \mathbb{R}$ is given by:
	\begin{equation}\label{eq:area}
		s_f(\Xi) := \frac{\alpha + \hat{\alpha}}{\frac{\alpha}{\Xi^2} + \hat{\alpha}}
	\end{equation}
	The model assumes that the chip area is normalized to one area unit. The non-ideality factor $\alpha$ denotes, how well the technology scales quadratically. The base technology area offset $\hat{\alpha}$ is dominated by components which do not scale. Both must be evaluated for the used set of technology nodes. Therefore, a small circuit with typical properties is synthesized, such as a basic router model (see Sec.~\ref{sec:results:modelfit}). Then, the parameters can be estimated using function fitting. In an ideal setting, $\alpha = 1$ and $\hat{\alpha} = 0$. As an example, we consider two layers implemented in an ideal theoretical $\tau(1) = \unit[45]{nm}$ and $\tau(2) = \unit[14]{nm}$ technology. The technology scaling factor is $s_f(\Xi(45,14)) = 10.2$. Between \unit[28]{nm} and \unit[45]{nm} nodes it is $s_f(\Xi(45,28)) = 2.58$. 
\end{mydef}

\subsubsection{Number of routers}\label{sec:routersPerLayer}
The different technology nodes influence not only the size of routers, but also their number per layer. The scaling factor $s_f$ can also be applied here to approximate a lower bound for the number of additional routers that can be implemented in a more advanced node. In that manner, we model a constant-area NoC per layer, which might not always be the most common integration approach (cf. our case study in Sec.~\ref{sec:results:casestudy}). If the area has been non-constant, the router count in faster layers would be reduced. Thus, the model underestimates advantages of our approach and therefore is valid, still.

\subsection{Timing model}

The transmission time of packets is determined by the individual timing of each router and the network topology. We model both characteristics; We consider clock delay of individual routers first and then deduct the propagation speed of packets traversing multiple routers.

\subsubsection{Clock delays} \label{sec:clockScalingFactor}
Routers in layers in mixed-signal nodes are potentially slower clocked whilst routers in the more advanced, digital technologies are clocked faster. The ratio, at which the clock delay in different technology nodes scales, is given by the clock scaling factor. There are two effects which influence the clock delay. It is larger than the interconnect delay for large technology nodes; it reduces with node scaling. Interconnect delay does not scale and therefore poses a limit for small nodes. Also, power constrains the maximum achievable clock frequencies. Therefore, the clock scaling factor is modeled by fitting a sigmoid function. Please note that this is an empirical and not a physical model. It has a high accuracy of the fit as shown in in Sec.~\ref{sec:results:modelfit}. If another (empirical or physical) model with similar accuracy is used, the results presented in this paper will not change.
 \begin{mydef}[Clock scaling factor]
	Let $\xi$ and $\iota$ be the indices of two chip layers with technologies  $\tau(\xi)$ and $\tau(\iota)$, with $\tau(\xi) > \tau(\iota)$  and with a relative technology scaling factor $\Xi (\xi, \iota)$. Let $c_b$ be the base clock delay of the layer with index $\xi$ and $c_c$ be the minimum achievable clock delay. Let $\beta$ be the maximum speedup achievable: $\beta:=c_b/c_c$. The \textit{clock scaling factor} $c_f: (\mathbb{R}) \rightarrow \mathbb{R}$ is given by:
	\begin{equation}
		c_f (\Xi) := \frac{\beta}{1+\hat{\beta} \exp\left(-\tilde{\beta}\left(\Xi-\bar{\beta}\right)\right)}
	\end{equation}
	The function converges to the maximum achievable speedup $\beta$. The other parameters must be set by fitting the function to a set of synthesis results (see Sec.~\ref{sec:results:modelfit}).
\end{mydef}

\section{Modeling NoC communication in heterogeneous 3D SoCs}\label{sec:model:comm}

We model the horizontal and vertical communication separately, since different factors are relevant: Communication within a layer is synchronous while communication between layers is not always. Our models calculate latency, throughput and transmission speed under zero load.

\subsection{Horizontal communication}

The speed at which a packet is transmitted horizontally, at zero load, is called \textit{propagation speed}. The propagation speed differs with technology nodes, since the number of routers and the clock frequency of routers differ between layers. Within a layer, routers are synchronous. The \textit{propagation speed of a packet within a layer is given by the distance traveled divided by the packet latency}.
We measure the \textit{distance} that packets travel. All possible positions of routers in a 3D SoC are given by the set $P= \mathbb{R} \times \mathbb{R} \times [\ell]$.
The x- and y-coordinates are measured in [m]\footnote{``measured in [\dots]'' refers to SI-units; ``$[\ell]$'' to the set $\{1,\dots,\ell\}$.}. We use the notation that the symbols $p_x$, $p_y$ and $p_z$ denote the components of each position $p \in P$. Further, packets have a payload, which is modeled by the number of flits transmitted $l \in L = \mathbb{N}$. Together, the set of packets is given by $D = P\times P \times L$. Packets are transmitted from a current (source) position to a destination position. (Please note, that the current position refers to the location of the packet during transmission. This position changes over time and does not refer to the position the packet was initially injected at into the network.) This yields the definition of the horizontal transmission distance:
\begin{mydef}[Horizontal transmission distance]
	Let $\pi$ be a packet with $\pi = \left(p_1, p_2, l\right)$, with source node $p_1$, destination node $p_2$ and length $l$. The horizontal transmission distance $s(\pi)$ is defined as the distance between source and destination positions in x- and y- dimension:
	\begin{equation}\label{eq:transmissionDistance}
	s(\pi) = \norm{\left(p_{1,x}, p_{1,y}\right) - \left(p_{2,x}, p_{2,y}\right)}
	\end{equation}
\end{mydef}
For example the distance between source and destination position in x- and y-dimension of a packet $\pi = (p_1, p_2, l)$ is calculated by $s(\pi) := \norm{\left(p_{1,x}, p_{1,y}\right) - \left(p_{2,x}, p_{2,y}\right)}_1$  in a mesh topology. The norm $\norm{\cdot}_1$ denotes the Manhattan norm ($\norm{p}_1 = \sum_{i = 1}^{n} p_n$ for $p\in \R^n$).

The \textit{latency} of a packet is calculated by the cumulative latency each router adds along the path. Each router requires $\delta(\xi)$ clock cycles to process the head flit in the layer $\xi \in [\ell]$. Thereafter, one flit is transmitted each clock cycle until end of packet. A single router finishes the transmission of a single packet with $l$ flits after $\delta(\xi) + l$ cycles. The constant $\rho(\xi)$ is defined as the average distance between routers in the layer $\xi$. Hence, a packet traverses $\nicefrac{s(\pi)}{ \rho(\xi)} + 1$ routers including the destination router.  This is illustrated for an example in Fig.~\ref{fig:commHor}, in which two consecutive packets are transmitted. In the example, routers have a head delay of $\delta = 3$ and pipelining $\chi = 2$. These considerations yield the following model for horizontal packet head latency that is accurate under assumption of zero load by construction.
\begin{mydef}[Horizontal packet head latency under zero load]
	Let $\pi$ be a packet with $\pi = (p_1, p_2, l)$ and $\xi \in [\ell]$ a layer. The average distance between routers in the layer $\xi$ is $\rho(\xi)$, \review{measured in [nm]}, and the delay for processing head flits per router is $\delta(\xi)$. The clock delay of routers is $\clk(\xi)$, measured in [s]. The \textit{horizontal packet head latency under zero load}, measured in [s], in layer $\xi$ is
	\begin{equation}\label{eq:packetTransmissionTime}
	\Delta_H(\pi, \xi) = \left(\frac{s(\pi)}{ \rho(\xi)}+1\right) \delta(\xi) \clk(\xi).
	\end{equation}
	\review{As given in Definition 4, the horizontal transmission distance is measured in [nm], but not in number of hops. Since the horizontal packet head latency is calculated from the number of hops passed by a packet, the horizontal transmission distance is multiplied with the average distance between routers. This yields the number of routers passed. We use average numbers, as routers will not be spaced evenly if the size of processing elements varies.} Furthermore, please note, that this model is accurate under zero load by construction. We verified this using simulations, as shown in Sec.~\ref{sec:results:modelfit} (Figs.~\ref{fig:SpeedupAlgo1} and \ref{fig:SpeedupAlgo2}).
\end{mydef}

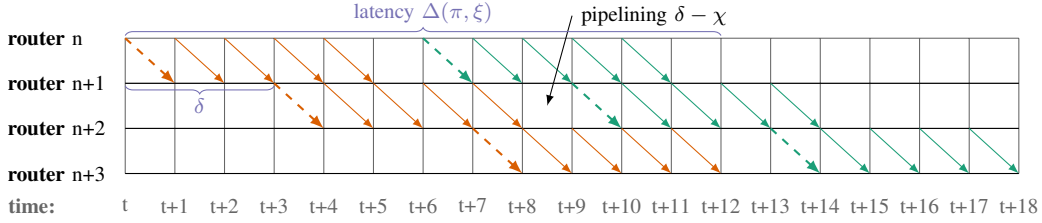
\begin{figure*}
	\centering

\newcommand{\drawpacket}[4]{
\draw[-latex, #4, thick, dashed, #3] (#1,#2) -- (#1+1, #2+1);
\foreach \i in {1, ..., 3}{
\draw[-latex, #4, #3] (#1+\i,#2) -- (#1+\i+1, #2+1);
}
\draw[-latex, #4, #3]  (#1+4,#2) -- (#1+5, #2+1);
}


	\begin{tikzpicture}[scale = 0.60, yscale=-1, xscale = 1.1,
cube/.style={thick, black, col1, fill = col3, fill opacity = 0.2},
axis/.style={-latex,col2, thick}, digital/.style={thick, col1}, digitalLines/.style = {col1}]

\draw[step=1.0,black!60,thin] (0,0) grid (18,3);
\foreach \i in {1,...,3}{
	\node[desc, anchor = east, align = right, xshift = -4pt] at (0, \i)   (a) {\textbf{router} n+\i};
	\draw[] (0,\i) -- (18,\i);
}
\node[desc, anchor = east, align = right, xshift = -4pt] at (0, 0)   (a) {\textbf{router} n\phantom{+0}};

\foreach \i in {1,...,18}{
	\node[yshift = -6pt,desc, anchor = north, align = center, black!60] at (\i, 3)   (a) {t+\i};
}
\node[yshift = -6pt,desc, anchor = north, align = center, black!60] at (0, 3)   (a) {t};
\node[yshift = -6pt,desc, anchor = north, align = center, black!60] at (-1.88, 3)   (a) {\textbf{time:}};

\drawpacket{0}{0}{}{col2}
\drawpacket{3}{1}{}{col2}
\drawpacket{7}{2}{}{col2}

\drawpacket{6}{0}{}{col1}
\drawpacket{9}{1}{}{col1}
\drawpacket{13}{2}{}{col1}

\node [desc, anchor = west] at (9,-.5)(pipe) {pipelining $\delta - \chi$};
\draw[-latex] (pipe.west) -- (8.5, 1.5);

\draw [decorate,decoration={mirror,brace,amplitude=3pt	,raise=1pt},yshift=0pt, draw = col3]
(0,1) -- (3,1) node [desc,col3,midway, xshift = 0pt,yshift = -2pt, anchor = north] {$\delta$};

\draw [decorate,decoration={brace,amplitude=3pt	,raise=1pt},yshift=0pt, draw = col3]
(0,0) -- (12,0) node [desc,col3,midway, xshift = 0pt,yshift = 2pt, anchor = south] {latency $\Delta(\pi, \xi)$};

\end{tikzpicture}

	\caption{Exemplary horizontal communication of two consecutive packets (orange, green).}
	\label{fig:commHor}
\end{figure*}

\begin{mydef}[Horizontal router throughput]
	Let $\pi$ be a packet with  $\pi = (p_1, p_2, l)$ and $\xi\in[\ell]$ a layer. The delay for processing head flits per router is $\delta(\xi)$. The router is pipelined with $\chi(\xi)$ $\in [0,\delta(\xi)]$ steps. The clock delay of routers is $\clk(\xi)$, measured in [s]. The \textit{horizontal router throughput}, measured in [flits/s], is given by the number of flits that a router can pass in a period of time:
	\begin{equation}\label{eq:latencyHorizontal}
		\hat{\Delta}_H(\pi, \xi) = \frac{l}{\left(l+\delta(\xi) - \chi(\xi)\right) \clk(\xi)}
	\end{equation}
\end{mydef}

\subsection{Vertical communication}\label{sec:commVertical}

Only vertical communication is effected by varying clock speeds. We model a non-purely synchronous communication, which allows to model different router and link architectures, such as the mesosynchronous proposed in Sec.~\ref{sec:architectures}.

\begin{mydef}[Vertical packet head latency under zero load]
	Let $\pi$ be a packet with $\pi = (p_1, p_2, l)$ and $\xi$ and $\lambda \in [\ell]$ layers with $p_{1_z} = \xi$ and $p_{2_z} = \lambda$. Without loss of generality, assume that $\xi \leq \lambda$. The clock delay of routers is $\clk(i)$ for all layers $i \in [\ell]$, measured in [s]. The \textit{vertical packet head latency under zero load (downwards)}, measured in [s], is given by the delay each router adds during head flit processing
	\begin{equation}\label{eq:latencyVdown}
	\Delta_V^{\downarrow}(\pi, \xi, \lambda) = \sum_{i = \xi}^{\lambda} \delta(i)\clk(i).
	\end{equation}
	The \textit{vertical packet head latency under zero load (upwards)}, measured in [s], is given by the delay each router adds during head flit processing plus a clock cycle for synchronization. This occurs only once during the path of the packet, since only two types of technology nodes are combined. The slower clock frequency dominates. This is illustrated in Fig.~\ref{fig:verticalTransmission} following the dashed thick arrow for the transmission of the head flit. In the Figure, the example uses routers in two layers, clocked at a frequency of $1$ and of $\nicefrac{1}{2}$. All routers have $\delta = 0$ and pipelining $\chi = 0$.
	\begin{equation}\label{eq:latencyVup}
	\Delta_V^{\uparrow}(\pi, \xi, \lambda) = \sum_{i = \lambda}^{\xi} \delta(i)\clk(i) + \clk(\xi).
	\end{equation}
	Please note, that this model, again, is accurate under zero load by construction, (cf. Sec.~\ref{sec:results:modelfit}).
\end{mydef}

\begin{mydef}[Vertical router throughput]
	Let $\pi$ be a packet with $\pi = (p_1, p_2, l)$ and $\xi$ and $\lambda \in [\ell]$ layers with $p_{1_z} = \xi$ and $p_{2_z} = \lambda$. Without loss of generality, assume that $\xi \leq \lambda$. Routers are pipelined with $\chi(i) \in [0,\delta(\xi)]$ steps in each layer $i \in [\ell]$. The clock delay of routers is $\clk(i)$, measured in [s]. The \textit{vertical router throughput}, measured in [flits/s], is given by the slowest router:
	\begin{equation}\label{eq:throughputV}
	\hat{\Delta}_V(\pi, \xi, \lambda) = \min_{i \in [\xi, \dots, \lambda]} \left\{ \hat{\Delta}(\pi, i) \right\}
	\end{equation}
	Long delays for processing a head flit are not relevant in the case of pipelining. Fig.~\ref{fig:verticalTransmission} demonstrates that the slowest clock dominates the throughput of the transmission for asynchronous chips using an exemplary two-layer chip with routers clocked at a frequency of $1$ and of $\nicefrac{1}{2}$.
\end{mydef}
\begin{figure*}
	\centering

\newcommand{\drawpacketClock}[6]{
\draw[-latex, #4, thick, dashed, #3] (#1,#2) -- (#1+#5, #2+1);
\foreach \i  [evaluate=\i as \j using \i*#6]  in {1,..., 4}{
\draw[-latex, #4, #3] (#1+\j,#2) -- (#1+\j+#5, #2+1);
}
}

\newcommand{\drawpacketClockUp}[6]{
	\draw[-latex, #4, thick, dashed, #3] (#1,#2) -- (#1+#5, #2-1);
	\foreach \i  [evaluate=\i as \j using \i*#6]  in {1,..., 4}{
		\draw[-latex, #4, #3] (#1+\j,#2) -- (#1+\j+#5, #2-1);
	}
}


	\begin{tikzpicture}[scale = 0.6, yscale=-1,xscale = 1.1,
cube/.style={thick, black, col1, fill = col3, fill opacity = 0.2},
axis/.style={-latex,col2, thick}, digital/.style={thick, col1}, digitalLines/.style = {col1}]

\foreach \i in {0,1,...,12}{
	\draw[black!60, thin] (\i, 0) -- (\i, 1);
}
\foreach \i in {0,0.5,...,12}{
	\draw[black!60, thin] (\i, 1) -- (\i,2);
}
\foreach \i in {0,...,2}{
	\draw[] (0,\i) -- (12,\i);
}
\node[desc, anchor = west, align = left, xshift = -4pt] at (-2.5, 0.5)   (a) {\textbf{slower layer}};
\node[desc, anchor = west, align = left, xshift = -4pt] at (-2.5,1.5)   (a) {\textbf{faster layer}};

\foreach \i in {1,2, ...,12}{
	\node[yshift = -0pt,desc, anchor = south, align = center, black!60] at (\i, 0)   (a) {t+\i};
}
\node[yshift = -0pt,desc, anchor = south, align = center, black!60] at (0, 0)   (a) {t};
\node[yshift = -0pt,desc, anchor = south, align = left, black!60] at (-2.07, 0)   (a) {\textbf{time:}};

\drawpacketClock{0}{0}{}{col2}{1}{1}
\drawpacketClock{1}{1}{}{col2}{.5}{1}

\drawpacketClockUp{7}{1}{}{col2}{1}{1}
\drawpacketClockUp{6}{2}{}{col2}{.5}{.5}

\draw [decorate,decoration={mirror,brace,amplitude=3pt	,raise=1pt},yshift=0pt, draw = col3]
(0,2) -- (2.5,2) node [desc,col3,midway, xshift = 0pt,yshift = -2pt, anchor = north] {vertical delay};

\draw [decorate,decoration={mirror,brace,amplitude=3pt	,raise=1pt},yshift=0pt, draw = col3]
(6,2) -- (8,2) node [desc,col3,midway, xshift = 0pt,yshift = -2pt, anchor = north] {vertical delay};

\draw [decorate,decoration={mirror,brace,amplitude=3pt	,raise=1pt},yshift=8pt, draw = col1]
(1,2.4) -- (5.5,2.4);
\draw [decorate,decoration={mirror,brace,amplitude=3pt	,raise=1pt},yshift=8pt, draw = col1]
(7,2.4) -- (12,2.4);
\node [desc,col1,xshift = 0pt,yshift = -2pt, anchor = north] at (6,2.7) {throughput dominated by slowest clock frequency};


\end{tikzpicture}

	\caption{Vertical communication is dominated by the slowest clock frequency.}
	\label{fig:verticalTransmission}
\end{figure*}
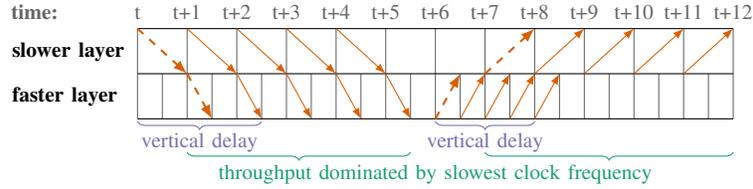

\section{Integration issues for heterogeneous 3D interconnects}\label{sec:potentials}

\review{Limitations of heterogeneous 3D interconnects are a result of different transmission speeds in varying technologies as found in Definitions 2 and 3}. This can be overcome by routing algorithms to achieve latency reductions. Routers are not purely synchronous, which will influence the throughput of routers along the packet's path, if it traverses multiple layers. This can be overcome by router architectures with increased throughput. Only simultaneous consideration of latency and throughput enables development of efficient interconnects, which impressively demonstrates the essential need for a co-design of routing strategies and router architectures in heterogeneous 3D SoCs.

\subsection{Tackling latency limitations via novel routing strategies}

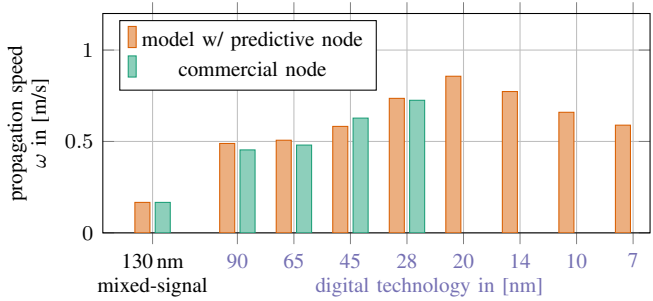
\begin{figure}
%
\pgfplotstableread[col sep = comma]{tables/omega.csv}\omegaData

\pgfplotsset{compat=1.11,
	/pgfplots/ybar legend/.style={
		/pgfplots/legend image code/.code={%
			\draw[##1,/tikz/.cd,yshift=-0.25em]
			(0cm,0cm) rectangle (3pt,0.8em);},
	},
}

\begin{tikzpicture}
				\begin{axis}[ybar,
				bar width=0.2cm,
				grid = both,
				ymin = 0,
				ymax = 1.2,
				xmax = 10.3,
								xtick={1.5,3, 4,5, 6, 7, 8, 9, 10},
				xticklabels={\unit[130]{nm}\\mixed-signal,\textcolor{col3}{90}, \textcolor{col3}{65}, \textcolor{col3}{45}, \textcolor{col3}{28}, \textcolor{col3}{20}, \textcolor{col3}{14}, \textcolor{col3}{10}, \textcolor{col3}{7}},
				xlabel= {\textcolor{col3}{digital technology in [nm]}}, 
				legend pos=north west ,
				height =4.5cm,
				width=\linewidth,
				ylabel style={align=center,desc, yshift = -1pt},
				xlabel style={desc, yshift = 14pt,xshift = 20pt, align=center},
				legend style={desc},
				xticklabel style={desc, align = center},
				yticklabel style={desc},
				ylabel={propagation speed \\ $\omega$ in [m/s]}
				]
				\addplot [col2, fill=col2!50] table[x=position,y=omegaModel] {\omegaData};
				\addplot [col1, fill=col1!50] table[x=position,y=omegaReal] {\omegaData};
				
				\legend{model w/ predictive node, commercial node};
				\end{axis}
		\end{tikzpicture}
	\vspace{-13pt}
	\caption{Propagation speed $\omega$ using a three-cycle router.}
	\label{fig:propagationSpeed}
\end{figure}

This publication answers whether communication via certain layers in heterogeneous 3D SoCs is faster, depending on the technology constraints, which can be exploited by routing algorithms. Intuitively, the first guess is that more advanced technology nodes are faster: routers certainly have a faster clock frequency. But there is a powerful adversary: the size of individual routers shrinks with better technology nodes. Thus more routers are located along the path of packet which add delay. To give a comprehensive answer, the proposed area and the proposed timing model must be considered simultaneously. Using Eq.~\ref{eq:transmissionDistance} and Eq.~\ref{eq:packetTransmissionTime}, and derivation, yield the propagation speed of a packet under zero load.
\begin{mydef}[Propagation speed]
	Let  $\xi\in[\ell]$ be a layer. The \textit{propagation speed in layer $\xi$} is
	\begin{equation}\label{eq:propagationSpeed}
	\omega(\xi) = \frac{\rho(\xi)}{\delta(\xi)\clk(\xi)}
	\end{equation}
	measured in [m/s]. It can be obtained by considering any packet $\pi$ with $\pi = (p_1, p_2, l)$ with distance $s(\pi)$. The speed is distance per time, i.e.\ $\omega(\xi) = \frac{s(\pi)}{\Delta_H(\pi, \xi) } $.
\end{mydef}
The propagation speed $\omega$ is shown in Fig.~\ref{fig:propagationSpeed} for commercial \unit[130]{nm} mixed-signal and \unit[90]{nm} -- \unit[28]{nm} digital technology, using the synthesis results for our NoC router with a head flit delay of  $\delta = 3$ and a 2$\times$2 NoC in the mixed-signal layer. This yields a horizontal transmission speed improvement of between $2.7\times$ and $4.3\times$, comparing mixed-signal and digital technologies. We see that the more powerful adversary which dominates is the clock scaling, whose influence is stronger than the effect of area scaling. To show the effect for other technology nodes, we fit the proposed models to the data (see Sec.~\ref{sec:results:modelfit}). The results are shown in Fig.~\ref{fig:propagationSpeed}, as well. The models can be used to predict the propagation speed for technology nodes below \unit[28]{nm}. This demonstrates potentials of our approach for more modern technologies, but we do not use this for the further evaluation, since it is predictive. The performance speed improvement is between $5.1\times$ and $3.3\times$. It is lower for more modern technologies due to limits posed by clock frequency scaling. Thus, clock frequency scaling remains dominant over area scaling, yet its advantages decline; Routing algorithms utilizing this are proposed in Sec.~\ref{sec:routing}.

\subsection{Tackling throughput limitations via novel router architectures}\label{sec:throughput}

We consider the influence of heterogeneity on throughput. Let's consider, for sake of simplicity, only packets with length $l$. Then, according to Eq.~\ref{eq:latencyHorizontal}, the throughput of horizontal communication is  $\hat{\Delta}_H = \frac{1}{\clk(\xi)}$: it is determined by the layer's clock frequency. If communication spans layers in another technology (i.e.\ with another clock frequency), Eq.~\ref{eq:throughputV} yields the vertical throughput:
\begin{align}
\begin{split}
\hat{\Delta}(\pi, \lambda) = &\min \{ \hat{\Delta}_V(\pi, \xi, \lambda), \hat{\Delta} (\pi, \lambda) \}\\
=&\hat{\Delta}_V(\pi, \xi, \lambda)  \leq \frac{1}{\clk(\xi)}
\end{split}
\end{align}
We have thereby shown that the throughput of packets which spans heterogeneous layers is determined by the slowest clock frequency: \emph{the chain is only as strong as its weakest link.} This effect poses a universal limitation to routing in heterogeneous 3D SoCs: \emph{communication may not span slower clocked layers if high throughput is required}. This issue cannot be circumvented by routing algorithms, since the only viable option is to avoid slower layers, which is impossible for a packet to and from this layer. This has two consequences: First, horizontal transmission in slower layers must be reduced to a minimum. Second, if a packet originates from a slow layer or is designated to a slow layer, the effects of their slow clock frequency must be minimized. This can only be achieved by novel router architectures; We propose an exemplary implementation in Sec.\,\ref{sec:architectures}.

\section{Tackling latency via routing strategies}\label{sec:routing}

In this section, routing strategies for heterogeneous 3D interconnects are developed. We start by abstracting the findings of our models into principles in Sec.~\ref{sec:principles}. Next, we shorty introduce some technical preliminary considerations for our setting in Sec.~\ref{sec:routing:prelim}. Finally, we can develop our routing strategies based on the principles in Secs.~\ref{sec:ra:1} and \ref{sec:ra:2}. The validity of the routing strategies is explained in Sec.~\ref{sec:deadlocks} by proving deadlock and livelock freedom.

\subsection{Principles for routing in heterogeneous 3D interconnects}\label{sec:principles}

The potentials as discussed in Sec.~\ref{sec:potentials} reveal that transmission through different layers can yield a performance advantage which is unique to heterogeneous 3D interconnects. This can be exploited by the following two paradigms for routing strategies:

\begin{itemize}[noitemsep,topsep=0pt, label={--}]
\item\textit{"Stay in faster layers!"}: Packets should stay as long as possible in layers which provide higher propagation speeds. An example is shown in Fig.~\ref{fig:stayinfasterlayer}. The sectional drawing of a two-layered chip is depicted. The layers are in MS and digital technology with $s_f = 4$. Usually, the data transmitted from routers R\,1 to R\,2 stay in the upper layer until reaching the router above R\,2 (depicted in orange color). This path is slower than the way back via the lower layer in the more advanced technology node. Thus, it is favorable to route all packets via the preferred path, depicted in green.

\item\textit{"Go through faster layers!"} If the performance gain is large, packets can  be routed via adjacent, faster layers since the path is faster. An example is shown in Fig.~\ref{fig:gothroughfasterlayer}. A sectional drawing of a two-layered chip is depicted. The layers are in mixed-signal and digital technology with $s_f = 4$. The routers R\,1 and R\,2 are communicating. Usually, data is transmitted via the upper layer, which is slower than the lower layer. Therefore, it is favorable to route packets via the orange path.
\end{itemize}

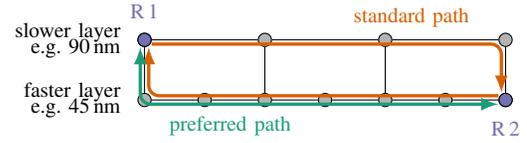
\begin{figure}
		\centering
		\begin{tikzpicture}[scale=.80]
	\makeNoCSectional{slower layer \\[-3pt]e.g. \unit[90]{nm}}{faster layer\\[-3pt]e.g. \unit[45]{nm}};
	
	\node [router, fill = col3]  (red1) at (01) {};
	\node [router, fill = col3]  (red2) at (62) {};
		
	\node [desc, above = 2pt of 01.north, col3] (pe1) {R\,1};
	\node [desc, below = 2pt of 62.south, col3] (pe2) {R\,2};
	
	\draw[-latex, col2, line width = 1.25pt, rounded corners] ([yshift = -2pt]01.east) -| ([xshift = -2pt]62.north) node [midway, desc, col2, yshift = 10pt, xshift = -34pt] {standard path};
	
	\draw[latex-, col2, line width = 1.25pt, rounded corners] ([xshift = 2pt]01.south) |- ([yshift = 2pt]62.west);
	
	\draw[latex-latex, col1, rounded corners, line width = 1.25pt] ([xshift = -2pt]01.south) |- node [midway, desc, col1, yshift = -8pt, xshift = 34pt] {preferred path} ([yshift=-2pt]62.west) ;
\end{tikzpicture}
		\caption{\textit{"Stay in faster layers!"}: The green paths are faster than the orange paths.}
		\label{fig:stayinfasterlayer}
\end{figure}
\begin{figure}
		\centering
		\begin{tikzpicture}[scale=.8]
	\makeNoCSectional{slower layer \\[-3pt]e.g. \unit[90]{nm}}{faster layer\\[-3pt]e.g. \unit[45]{nm}};
	
	\node [router, fill = col3]  (col3) at (01) {};
	\node [router, fill = col3]  (col3) at (31) {};
		
	\node [desc, above = 2pt of 01.north, col3] (pe1) {R\,1};
	\node [desc, above = 2pt of 31.north, col3] (pe2) {R\,2};
	
	\draw[latex-latex, col2, line width = 1.25pt] ([yshift = 2pt]01.east) -- ([yshift = 2pt]31.west) node [midway, desc, col2, yshift = 6pt] {standard path};
	
	\draw[latex-latex, col1, rounded corners, line width = 1.25pt] ([xshift = 2pt]01.south) -- ([xshift = 2pt, yshift = 2pt]02.east) --node [midway, desc, col1, yshift = -10pt] {preferred path} ([xshift=-2pt, yshift = 2pt]62.west) -- ([xshift=-2pt]31.south) ;
	
\end{tikzpicture}
		\caption{ \textit{"Go through faster layers!"}: The green path from R\,1 to R\,2 is longer yet faster than the orange path.}
		\label{fig:gothroughfasterlayer}
\end{figure}
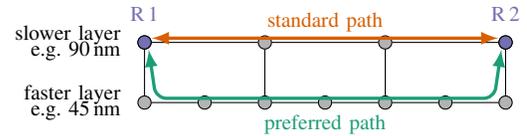

We apply the two aforementioned paradigms to develop two exemplary routing algorithms. The proposed models provide relevant information on their potentials and to set parameters of the routing algorithms. Our proposed models allow to assess which routings are applicable and under which circumstances, since the models are generally valid, i.e.\ can be applied to any topology and set of technology parameters (beyond the proposed algorithms and the setting). Thus, we do not lose generality of models, yet demonstrate their expressiveness.

\subsection{Preliminary considerations}\label{sec:routing:prelim}

\subsubsection{Setting}

A heterogeneous 3D SoC with $\ell \in \N$ layers is used. Its layers are ordered by technology node, as in the vast majority of works on 3D SoCs, e.g.\ \cite{Chen.2016}. The most coarse-grained technology is at the top whilst the most fine-grained technology is bottom-most. Reordering the layers does not influence the models and principles and only requires minor changes to the proposed routing algorithms; hence, this does not lead to a loss of generality. But the order reduces the complexity of descriptions. Our approach is applicable to scenarios without ordered layers, with minor modifications.

Within the heterogeneous 3D SoC we implement a 3D NoC. Each layer has a grid with $m_\xi$ rows and $n_\xi$ columns, wherein $\xi \in [\ell]$ is the layer index. Routers are disposed in rows and columns. Neighboring routers are connected horizontally forming a $m_\xi$-$n_\xi$-mesh topology in layers, which is the most common NoC topology. No router has more than one link in the same direction, e.g.\ we do not model long range links. All routers, except those on the bottommost layer, have a (bidirectional) vertical link to the adjacent router in the next lower layer. This is possible thanks to the ordering of layers (cf.\ Fig.~\ref{fig:finemeshed}). The set of routers $V$ is also the vertex set of the network digraph $T = (V,A)$.\footnote{In \textit{Duato} \cite{Duato.1993} the network digraph is called \emph{interconnection network}.} The set of arcs $A$ contains the directed links between routers.

\subsubsection{Addresses in the network}

Locations of routers are given by a coordinate system with its origin in the SoC's top left corner, as shown in Fig.~\ref{fig:neswud}. Routers have both a physical location and a row and column number. The implementation of routing algorithms must be efficient, i.e.\ calculating with the physical locations is not realistic; using row, column, and layer numbers is. Rows and columns are based on the network digraph and not the physical locations: For example, pairs of neighbored routers in adjacent layers do not necessarily  have the same physical x- and y- coordinate but the same column and row number. This is shown in Fig.~\ref{fig:finemeshed}. We do not depict this in all figures for sake of simplicity. If routers are depicted  as stacked (cf.\ Fig.~\ref{fig:stayinfasterlayer}), we will intend a placement comparable to Fig.~\ref{fig:finemeshed}. We use the notation $w = (w_x, w_y, w_z)$ for $w \in W = \N^3$, which determines row, column, and layer of each router, which is equivalent to the address. An injective function  $m: W \rightarrow P $ converts addresses to locations of routers. Packets with source and destination address are given by $\tilde{D} = W \times W \times L$.

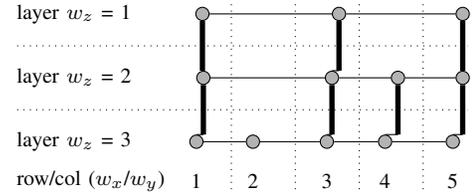
\begin{figure}
		\centering

\newcommand{\tsv}[2]{\draw[line width = 2pt] (#1.south) -- (#1.south |- #2.north);
	\draw (#1.south |- #2.north) -- (#2.north);}
	\begin{tikzpicture}[scale=.85]
	
	\foreach \x in {1,3,5}{
		\node [router]  (\x1) at (\x+.4*rnd-.2,0) {};}
	\foreach \x in {1,3,4,5}
	{\node [router]  (\x2) at (\x+.4*rnd-.2,-1) {};}
	\foreach \x in {1,2, 3,4,5}
	{\node [router]  (\x3) at (\x+.4*rnd-.2,-2) {};}

	\draw(11) -- (31) -- (51);
	\draw(12) -- (32) --(42) -- (52);
	\draw(13) -- (23) --(33) -- (43) --  (53);
	\tsv{11}{12}
	\tsv{12}{13}
	\tsv{31}{32}
	\tsv{51}{52}
	\tsv{32}{33}
	\tsv{42}{43}
	\tsv{52}{53}
	
	\node[desc, below =6pt of  13] (row1){1};
	\node[desc, below =6pt of  23] (row2){2};
	\node[desc, below =6pt of  33] (row3){3};
	\node[desc, below =6pt of  43] (row4){4};
	\node[desc, below =6pt of  53] (row5){5};
	\node[desc, anchor = west, align = left] (layer1) at (-2,0) {layer $w_z$ = 1};
	\node[desc, anchor = west, align = left] (layer2) at (-2,-1) {layer $w_z$ = 2};
	\node[desc, anchor = west, align = left] (layer3) at (-2,-2) {layer $w_z$ = 3};
	\node[desc, anchor =west, align = left] (row) at (layer1.west |- row1) {row/col ($w_x$/$w_y$)};
	
	\foreach \x in {1,2, 3,4}{
 	\draw[dotted] (\x+.5, .25) -- (\x+.5, -2.75);}
 
 	\foreach \y in {0,1}{
 	\draw[dotted] (-1.85,-\y-.5) -- (5.2, -\y-.5); }

	\end{tikzpicture}
		\vspace{-11pt}
		\caption{Logical order using redistribution.}
		\label{fig:finemeshed}
\end{figure}
\begin{figure}
	\centering
	\begin{tikzpicture}
	[cube/.style={thick, black, col1, fill = col3, fill opacity = 0.1},
	axis/.style={-latex,col2, thick}, direction/.style={black, col3,  -latex}]
	
	draw the axes
	\draw[axis] (0,0,0) -- (3,0,0) node[desc, anchor=west]{$v_x \in \mathbb{N}$};
	\draw[axis] (0,0,0) -- (0,-3,0) node[ desc, anchor=west]{$v_z \in [\ell]$};
	\draw[axis] (0,0,0) -- (0,0,3) node[desc, anchor=east]{$v_y \in \mathbb{N}$};
	
	\draw[direction] (1, -1, 1) -- (1, -2, 1) node[desc, anchor = west] {down};
	\draw[cube] (0,-1,0) -- (2,-1,0) -- (2,-1,2) -- (0,-1,2) -- cycle;
	
	\draw[direction] (1, -1, 1) -- (1, 0, 1) node[desc, anchor = west] {up};
	
	\draw[direction] (1, -1, 1) -- (1.9, -1, 1) node[desc, anchor = west, xshift = .3] {east};
	\draw[direction] (1, -1, 1) -- (0.1, -1, 1) node[desc, anchor = east, xshift = -1.3] {west};
	\draw[direction] (1, -1, 1) -- (1, -1, 1.9) node[desc, anchor = north, yshift = 1, xshift = -1.5] {south};
	\draw[direction] (1, -1, 1) -- (1, -1, 0.1) node[desc, anchor = south,  yshift = -1, xshift = 1.5] {north};
	
	\end{tikzpicture}
		\vspace{-11pt}
		\caption{Cardinal directions in model coordinates $W$.}
		\label{fig:neswud}
\end{figure}
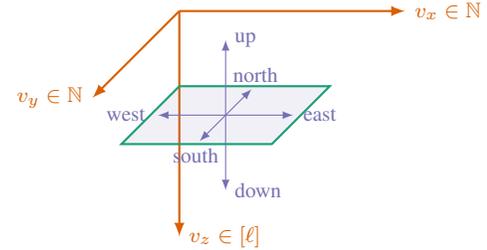


\subsubsection{Cardinal Directions}
We use the six cardinal directions $C := \{\north, \east, \south, $\newline$ \west, \up,\down\}$ to sort the arcs as shown in Fig.~\ref{fig:neswud}. We define functions which return the set of all links in one of these cardinal directions. These are given for all links $(v,w) \in A$:
\begin{alignat*}{3}
(v,w) &\in \north(A) &\qquad\Leftrightarrow\qquad&   v_x = w_x, v_y > w_y, v_z = w_z \\[-2pt]
(v,w) &\in \east(A) &\qquad\Leftrightarrow\qquad &  v_x < w_x, v_y = w_y, v_z = w_z\\[-2pt]
(v,w) &\in \south(A) &\qquad\Leftrightarrow\qquad &  v_x = w_x, v_y < w_y, v_z = w_z \\[-2pt]
(v,w) &\in \west(A) &\qquad\Leftrightarrow\qquad  & v_x > w_x, v_y = w_y, v_z = w_z\\[-2pt]
(v,w) &\in \up(A) &\qquad\Leftrightarrow\qquad  & v_z > w_z  \\[-2pt]
(v,w) &\in \down(A) &\qquad\Leftrightarrow\qquad&  v_z < w_z
\end{alignat*}
For example, $\north(A)$ contains all links pointing north. 
We further introduce functions that return neighbors of routers in a certain cardinal direction, if a link exists.\footnote{Note, that the above functions are only well defined, if no router has more than one link to the same direction.} Routers at the edges of the network do not have links in that direction which is given by the value $\zero$. We define for all $f \in C$:
\begin{align*}
f: V &\rightarrow V \cup \{\zero\}  \\[-2pt]
v &\mapsto
\begin{cases}
w & \text{if } (v,w) \in f(A) \\
\zero & \text{otherwise} .
\end{cases}
\end{align*}


\subsection{Applying principle 1: \AlgorithmI\ - routing algorithm}\label{sec:ra:1}

We apply principle 1, \textit{"Stay in faster layers!"} and design a minimal and deterministic routing algorithm. Let $\tilde{\pi} = (v, w, l)$ be a packet. If the packet is not transmitted within a layer, i.e.\ $v_{z} \neq w_z$, the faster layer must be identified. Therefore, we apply Eq.~\ref{eq:propagationSpeed} to calculate the average propagation speed at design time. This yields the following rules for transmission of packet $\pi$ (in router with address $v$):
\begin{itemize}[noitemsep,topsep=0pt, label={--}]
	\item If ${\omega}(v_{z}) < {\omega}(w_{z})$,  XYZ routing will be applied.
	\item If ${\omega}(v_{z}) > {\omega}(w_{z})$,  ZXY routing will be applied.
	\item If ${\omega}(v_{z}) = {\omega}(w_{z})$, either will be selected at design time, depending on other network properties such as energy consumption of routers.
\end{itemize}
We call this routing algorithm \AlgorithmI.\footnote{Minimality refers to the shortest path in the interconnection network. In terms of hop distance the proposed routing algorithm is not minimal. It is, however, if the links in the interconnection graph are weighted with their speed.} Since layers are ordered by technology and hence by transmission speed, the implementation 
extends deterministic XYZ simply by reordering \texttt{if}-statements.
Routers will only require additional flag storing information if faster layer is located below, above or is indeed this actual layer. The resulting routing is illustrated in Fig.~\ref{fig:routingCase2}.

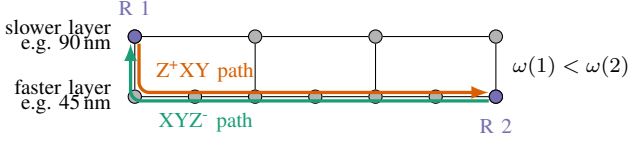
\begin{figure}
	\centering
	\begin{tikzpicture}[scale=.8]
	\makeNoCSectional{slower layer \\[-3pt]e.g. \unit[90]{nm}}{faster layer\\[-3pt]e.g. \unit[45]{nm}};
	
	\node [router, fill = col3]  (red1) at (01) {};
	\node [router, fill = col3]  (red2) at (62) {};
	\node [desc, above = 2pt of 01.north, col3] (pe1) {R 1};
	\node [desc, below = 2pt of 62.south, col3] (pe2) {R 2};

	\draw[-, arrows={-latex[col2,fill=col2]}, draw=col2, rounded corners, line width = 1.25pt] ([xshift=2pt]01.south) |- node [midway, desc, col2, yshift = 8pt, xshift = 25pt] {Z\textsuperscript{+}XY path}  ([yshift=2pt]62.west) ;

	\draw[latex-, col1, rounded corners, line width = 1.25pt] ([xshift = -2pt]01.south) |- node [midway, desc, col1, yshift = -8pt, xshift = 28.5pt] {XYZ\textsuperscript{-} path} ([yshift=-2pt]62.west) ;
	
	\node [desc, right of =31] (node1) {};
	\node [desc, right of =62] (node2) {};
	\node[desc] (omega) at ($(node1)!0.5!(node2)$){$\omega(1) < \omega(2)$};


\end{tikzpicture}
	\caption{\ZXYZ\ routing: transmission through the lower layer.}
	\label{fig:routingCase2}
\end{figure}
\begin{definition} [Routing function $R_1$ for \AlgorithmI\ routing]\label{def:R_1}
	Let $T = (V, A)$ be the topology digraph with the set of routers $V$ and the set of links $A$. Further, $\mathcal{P}(A)$ is the power set of $A$. The routing function $R_1: V \times V \rightarrow \mathcal{P}(A) $ is defined as:\footnote{Due to the setting all routers have a downwards vertical link (except those in the bottommost layer); thus $\{\zero\}$ is impossible by construction (proved in Lemma~\ref{lem:con1}).}
	\begin{align*} 
	(v,d) &\mapsto
	\begin{cases}
	\emptyset & \text{for } v = d \\[-2pt]
	\{\north(v)\} & \text{for } v_x = d_x, v_y > d_y, v_z \geq d_z \\[-2pt]
	\{\east(v)\} & \text{for } v_x < d_x, v_z \geq d_z \\[-2pt]
	\{\south(v)\} & \text{for } v_x = d_x, v_y < d_y, v_z \geq d_z \\[-2pt]
	\{\west(v)\} & \text{for } v_x > d_x, v_z \geq d_z \\[-2pt]
	\{\up(v)\} & \text{for } v_x = d_x, v_y = d_y, v_z > d_z \\[-2pt]
	\{\down(v)\} & \text{for } v_z < d_z.
	\end{cases}
	\end{align*}
\end{definition}

\subsection{Applying principle 2: \AlgorithmII\ - routing algorithm}\label{sec:ra:2}

We apply principle 2, \textit{"Go through faster layers!"}. This requires to identify a quicker path for packets using detours. The identification of the best path depends on the position of source and destination, since there is an overhead when routing to the fastest layer for vertical transmission. We assess under which circumstances routing via an adjacent layer is advantageous. Let $\tilde{\pi} = (v, w,l)$ be a packet with source address $v$ and destination address $w$. Let $\pi$ be the corresponding packet after applying $m$ to convert addresses to locations. The transmission time under zero load in the layer $v_z$ is $\Delta_H(\pi, v_z)$ (Eq.~\ref{eq:packetTransmissionTime}). Let $\lambda \in [\ell]$ be another layer, through which the packet could potentially be transmitted. The transmission time via layer $\lambda$ is the transmission time for traversing vertical links, plus time within layer $\lambda$. Applying the model yields the condition under which routing via layer $\lambda$ has a smaller latency:
\begin{equation}\label{eq:reroutingLatency}
\begin{split}
\Delta_H(\pi, \xi) \quad>\quad & \Delta_V^{\downarrow}(\pi,\xi, \lambda) + \Delta_H(\pi, \lambda) \\
+& \Delta_V^{\uparrow}(\pi,\xi, \lambda) - 2\delta(\lambda)\clk(\lambda)
\end{split}
\end{equation}

We calculate a threshold distance $\phi(\xi, \lambda)$ that determines the minimum distance in layer $\xi$ for which rerouting via layer $\lambda$ is faster. Please note that we assume two layers in disparate technologies which are adjacent. It is not useful to use another than the uppermost digital layer to save vertical transmission time. Nonadjacent layers in mixed signal nodes have larger thresholds. Eq.~\ref{eq:reroutingLatency}, with $\phi:=s(\pi)$ yields $\left( \frac{\phi\delta(\xi)}{\rho(\xi)} - \rho(\xi) - 1\right) \clk(\xi) = \left( \frac{\phi}{\rho(\lambda)} + 1 \right)\delta(\lambda) \clk(\lambda)$, which is transformed to:
\begin{align}
		\phi(\xi, \lambda) =
	\begin{cases}
	\frac{\left(\delta(\xi) \clk(\xi) + \delta(\lambda) \clk(\lambda) + \clk(\xi)\right)\rho(\xi)\rho(\lambda)}{\delta(\xi) \clk(\xi)\rho(\lambda) - \delta(\lambda)\clk(\lambda)\rho(\xi)}& \text{for } \xi < \lambda\\
	\infty & \text{else}
	\end{cases}\label{eq:phi}
\end{align}

Note, that $\infty$ can be replaced by any value larger the size of the chip. The two routing conditions are: (a) If a $\lambda$ exists with $s(\pi) = s((m(v), m(w), l)) > \phi(v_z, \lambda)$, ZXY routing will be applied in direction of $\argmin_{\lambda \in [\ell]}\phi(v_z, \lambda)$. (b) If $s(\pi) = s((m(v), m(w), l)) \leq\phi(v_z, \lambda)$ for all $\lambda \in [\ell]$, XYZ routing will be applied. There are two bottlenecks for run-time calculation: First, online selection of the best layer by evaluation of $\argmin$ is too expensive. A layer $\Lambda$ must be selected at design time. From a practical standpoint, the uppermost digital layer is preferred because it offers high speed and low overhead for vertical transmission.\footnote{Without loss of generality, we set $\Lambda:=\ell$ in proofs.} Second, addresses must be converted in locations. Therefore, we convert the location threshold distance $\phi$ into a hop distance by division through the average router distance in the digital layers:
\begin{equation}
  \Phi(\xi, \Lambda) :=\left\lceil\nicefrac{\phi(\xi, \Lambda)}{ \rho(\ell)}\right\rceil\label{eq:Phi}
\end{equation}
It is required that $\phi$ is smaller than the outside measurements of the chip so that the routing can be applied. For a combination of a commercial \unit[130]{nm} mixed signal node with commercial 90 -- \unit[28]{nm} digital nodes and a 4-4 NoC in the layer in mixed signal technology, $\phi$ is between 0.63 and 0.45 for a chip with edge length normalized to 1. Hence, packets traveling more than 2 or 3 hops in the layer in mixed signal node are routed via the adjacent layer.

To summarize, the routing algorithm has these conditions for a packet $\tilde{\pi} = (v,w,l)$ in router $v$:
\begin{itemize}[noitemsep,topsep=0pt, label={--}]
	\item If $|v_{x}-w_{x}| + |v_{y}-w_{y}| \leq \Phi (\xi, \Lambda)$, XYZ routing will be applied.
	\item If $|v_{x}-w_{x}| + |v_{y}-w_{y}|  > \Phi (\xi, \Lambda)$, the packet will be routed $\down$.
\end{itemize}
We call this routing \AlgorithmII. It is illustrated in Fig.~\ref{fig:reroutingCase2}. 
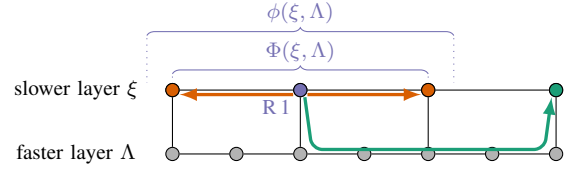
\begin{figure}
	\centering

\begin{tikzpicture}[scale=.85]
	\makeNoCSectional{}{};
	\node [desc, left of =01, xshift = -8pt] (node1) {slower layer $\xi$};
	\node [desc, left of =02, xshift = -8pt] (node2) {faster layer $\Lambda$};
	
	\node [router, fill = col3]  (col11) at (11) {};
	\node[router, fill = col2] (col01) at (01){};
	\node[router, fill = col2] (col21) at (21){};
	\node[router, fill = col1] (col31) at (31){};
		
	\node [desc, left = 0.05pt of 11.south, col3, yshift = -3.8pt] (pe1) {R\,1};
	
	\draw[-latex, col2, line width = 1.25pt] ([yshift = -2pt]11.west) -- ([yshift = -2pt]01.east) node [midway, desc, col2, yshift = 6pt] {};
	\draw[-latex, col2, line width = 1.25pt] ([yshift = -2pt]11.east) -- ([yshift = -2pt]21.west) node [midway, desc, col2, yshift = 6pt] {};
	\draw[-latex, col1, rounded corners, line width = 1.25pt] ([xshift = 2pt]11.south) -- ([xshift = 2pt, yshift = 2pt]22.east) --node [midway, desc, col1, yshift = 6pt] {} ([xshift=-2pt, yshift = 2pt]62.west) -- ([xshift=-2pt]31.south) ;
	
	\draw [decorate,decoration={brace,amplitude=3pt	,raise=1pt},yshift=-10pt, draw = col3]
	($ (01.north) + (0,0.1) $) -- ($ (21.north) + (0,0.1) $) node [desc,col3,midway, xshift = 0pt,yshift = 2pt, anchor = south] {$\Phi(\xi, \Lambda)$};
	\draw[dotted, col3] ($ (01.north) + (0,0.1) $)  -- (01.north);
	\draw[dotted, col3] ($ (21.north) + (0,0.1) $)  -- (21.north);
	
	\draw [decorate,decoration={brace,amplitude=3pt	,raise=1pt},yshift=-10pt, draw = col3]
	($ (01.north) + (-.4,0.7) $) -- ($ (21.north) + (.4,0.7) $) node [desc,col3,midway, xshift = 0pt,yshift = 2pt, anchor = south] {$\phi(\xi, \Lambda)$};	
	\draw[dotted, col3] ($ (01.north) + (-.4,0.7) $)  --  ($ (01.north) + (-.4,0) $);
	\draw[dotted, col3] ($ (21.north) +  (.4,0.7) $)  --  ($ (21.north) +  (.4,0) $);
\end{tikzpicture}
%
	\caption{\AlgorithmII\ routing: A detour is faster for long distances.}
	\label{fig:reroutingCase2}
\end{figure}


\begin{definition} [Routing function $R_2$ for \AlgorithmII] \label{def:R_2}
	Let $T = (V, A)$ be the topology digraph with the set of routers V and the set of links A. Let $\Lambda$ be a layer which is selected for rerouting at design time. Let $\Phi(\xi, \Lambda)$ be a threshold for rerouting according to Eq.~\ref{eq:Phi}. The routing function $R_2: V \times V \rightarrow \mathcal{P}(A) $ is defined as:
	\begin{align*}
	(v,d) &\mapsto
	\begin{cases}
	\emptyset & \text{for } v = d \\[-2pt]
	\{\down(v)\} & \text{for } |v_x-d_x|+|v_y + d+y| \\[-2pt]
						& \phantom{for } >\Phi(v_z, \Lambda), v_z\geq d_z\\[-2pt]
	\{\down(v)\} & \text{for } v_z < d_z. \\[-2pt]
	\{\north(v)\} & \text{for } v_x = d_x, v_y > d_y, v_z \geq d_z, \\[-2pt]
						& \phantom{for } |v_y-d_y|\leq\Phi(v_z, \Lambda) \\[-2pt]
	\{\east(v)\} & \text{for } v_x < d_x, v_z\geq d_z,\\[2pt]
						&  \phantom{for }|v_x-d_x|\leq\Phi(v_z, \Lambda) \\[-2pt]
	\{\south(v)\} & \text{for } v_x = d_x, v_y < d_y, v_z\geq d_z, \\[2pt]
						&  \phantom{for }|v_x-d_x|\leq\Phi(v_z, \Lambda)  \\[-2pt]
	\{\west(v)\} & \text{for } v_x > d_x, v_z\geq d_z, \\[2pt]
						&  \phantom{for }|v_x-d_x|\leq\Phi(v_z, \Lambda)  \\[-2pt]
	\{\up(v)\} & \text{for } v_x = d_x, v_y = d_y, v_z > d_z \\
	\end{cases}
	\end{align*}
\end{definition}

\subsection{Proof of validity: deadlock and livelock freedom}\label{sec:deadlocks}

We prove that the routing algorithms are free of deadlocks and livelocks. We make use of \textit{Duato's theorem} \cite{Duato.1993}, according to which a routing is deadlock-free if the routing function is connected and the channel dependency graph is cycle free. We also use terms and definitions from \cite{Duato.1993} without further explanation, such as \emph{routing function}, \emph{adaptive}, \emph{connected}, \emph{direct dependency}, and \emph{channel dependency graph}. If there is a direct dependency from $a$ to $b$, we also say: \emph{>>$b$ is direct dependent on $a$.<<} Graph related terms like \emph{path}, \emph{closed walk}, or \emph{cycle}  are used as defined in \cite{Korte.2002}. 

We introduce the terms \emph{possible turn} and \emph{impossible turn} according to a routing function $R$. These terms denote, if the routing functions permits consecutive flow of packets in these directions.
\begin{definition} A pair of cardinal directions $(f,g) \in C\times C$ is called a possible turn according to $R$, if there exist two consecutive arcs, $(u,v)$ and $(v,w) \in A$, with:
$(u,v) \in f(A)$, $(v,w) \in g(A)$
and there is a direct dependency from $(u,v)$ to $(v,w)$.
A pair of cardinal directions that is not a possible turn is called an impossible turn according to $R$.
\end{definition}
\begin{lemma}{}
	\label{lem:cyc}
	If there is a cycle in the channel dependency graph (CDG), then we can also find a closed walk $(v_1, a_1, v_2, \dots, v_k, a_k, v_1)$ (for $k \in \N$) in the topology digraph with
	\begin{itemize}[noitemsep,topsep=0pt, label={--}]
		\item $a_{i+1}$ is direct dependent on $a_i$ for all $i \in \{1, \dots, k-1\}$,
		\item and $a_1$ is direct dependent on $a_k$.
	\end{itemize}
\end{lemma}
\begin{proof}
	Assume that there is a cycle $(\{a_1, \dots, a_k\},$ $\{(a_1, a_2),\dots,(a_{k-1},a_k),(a_k,a_1)\})$ in the CDG. According to the definition of direct dependency, the destination node of $a_i$ in the topology digraph is also the source node of $a_{i+1}$ (for all $i \in \{1, \dots, k\}$, and $a_{k+1}:= a_1$). Let us call this node $v_{i+1}$ (for all $i \in \{1, \dots, k\}$). Then, $(v_{k+1}, a_1, v_2, \dots, v_k, a_k, v_{k+1})$ is a closed walk in the topology digraph.
\end{proof}


\subsection{\AlgorithmI: $R_1$ is deadlock-free} 

By looking at the definition of $R_1$, we can determine the impossible turns and the possible turns. Here, we assume that the numbers of rows, columns and layers $m_\xi$, $n_\xi$ and $\ell$ are not too small. We assume $m_\xi, n_\xi, \ell \geq 2$ for all $\xi \in \{1, \dots, \ell\}$ as a precaution. Table \ref{tbl:turns} shows which turns are possible.

\begin{lemma}
	\label{lem:link}
	When $R_1$ gives a direction, then the necessary link always exists.
\end{lemma}
\begin{proof}
	Places without links in some directions are: (a) At the outside faces of the 3D NoC cuboid links at edges of layers, upward links from the topmost layers, and downward links from the bottommost layer do not exist. (b) Some upward links do not exist between layers if one layer is in another technology than the other layer. 	
	a) By looking at the definition of $R_1$, one can check that every routing step brings the packet nearer to $d$. Hence, the nonexistent links on the outer border of the 3D-NoC are never taken by $R_1$.
	b) Not every router has an $\up$-link. Every router, except those in the bottommost layer, has a down link by premise. Downward links in a router are upward links in the router below:. When router $v$ has the same $x$- and $y$-coordinates as the destination router $d$ and $v$ is below $d$, $v$ has an $\up$-link. These are also the conditions for traveling $\up$ in $R_1$.
\end{proof}

{\footnotesize \sffamily
	\begin{table}
		\caption{Possible turns $(f,g)$ in $R_1$ and $R_2$.}
		\label{tbl:turns}
		\centering
		\begin{tabular}{|cc|cccccc|}
			\hline
			& \multicolumn{1}{c}{$g$:}& n. & e. & s. & w. & u. & d. \\
			\hline
			\multirow{6}{*}{$f$:} & n. & 1 & 0 & 0 & 0 & 1 & 0 \\
			&e. & 1 & 1 & 1 & 0 & 1 & 0 \\
			&s. & 0 & 0 & 1 & 0 & 1 & 0 \\
			&w. & 1 & 0 & 1 & 1 & 1 & 0 \\
			&u. & 0 & 0 & 0 & 0 & 1 & 0 \\
			&d. & 1 & 1 & 1 & 1 & 0 & 1 \\
			\hline
		\end{tabular}
	
	\end{table}
}

\begin{lemma}
	\label{lem:con1}
	$R_1$ is connected.
\end{lemma}

\begin{proof}
	
Let $s$ and $d$ be any two vertices in $V$. $R_1$ returns a direction for every vertex except $d$ (it returns $\emptyset$). The links in the chosen direction always exist (Lemma \ref{lem:link}). If we apply the routing function step by step and proceed through the network in the returned directions, we will find a route. As shown in the proof of livelock-freedom, the route is not infinite (Theorem \ref{thm:LLr1}). Hence, it terminates. Termination can only happen at $d$, by definition. Hence, with the routing function $R_1$, we always find a path from $s$ to $d$.
\end{proof}

\begin{theorem}
	\label{thm:r1}
	$R_1$ is deadlock-free.
\end{theorem}
\begin{proof}
	$R_1$ is connected, because of Lemma \ref{lem:con1}. Assume, that the CDG of $T$ and $R_1$ has a cycle. Lemma \ref{lem:cyc} proves that T has a cycle where each two consecutive arcs are direct dependent.\\	
	\emph{Case 1: All vertices of the cycle are in the same layer.} We know by \cite{Dally.1987} that XY routing has a cycle free CDG due to impossible turns. Thus, Case 1 does not occur.\\	
	\emph{ Case 2: The vertices of the cycle are in at least two different layers.} Since the vertices are in different layers, there is at least one arc, which goes up. According to table \ref{tbl:turns}, the only possible direction after >>up<< is >>up<< and the cycle could never be closed. Hence, Case 2 is also impossible.\\
	We have shown by contradiction that the CDG is cycle-free and apply Duato's Theorem on $R_1$.
\end{proof}

\subsection{\AlgorithmII: $R_2$ is deadlock-free}

Again, we can determine the set of possible turns. It can be seen in Table \ref{tbl:turns}.

\begin{lemma}
	\label{lem:con2}
	$R_2$ is connected.
\end{lemma}
\begin{proof}
	Let $s$ and $d$ be any two vertices in $V$. We construct a path $(s=v_{1},\dots,v_{k}=d)$ with $\ v_{i} \in V$ for all $i \in [k],\ k \in \mathbb{N}$ from $s$ to $d$ by using links $\left(c_{1} = \left(v_1, v_2\right),\dots, c_{k-1} = \left(v_{k-1}, v_k\right)\right)$  with $c_{i} \in A$ for all $i \in [k-1]$, which are consecutively delivered by the routing function $R_{2}$. 
	
	\emph{Case 1 (The source is above the destination ${s_{z} < d_{z}}$)}:	As in the proof of Lemma \ref{lem:con1}, the route starts with a sequence of $\down$s until the destination layer is reached. Now the routing goes as explained in Case 2.
	
	\emph{Case 2 (The source is below the destination or on the same layer ${s_{z} \geq d_{z}}$)}: The next links depend on the logical value of $||s-d|| \geq \Phi(s_z).$
	
	\emph{Case 2.1 ($||s-d|| \geq \Phi(s_z)$)}:	If the condition is true, the next link will be $\down$. The value of $||s-d||$ is the same as $||v_2-d||$. The value of $\Phi(z)$ is the same for all $z< \Lambda$. Hence, layer $\Lambda$ will be reached via a sequence of $\down$s. The rest of the path is constructed as in Case 2.2.
	
	\emph{Case 2.2 ($||s-d|| < \Phi(s_z)$)}:
	Here, $R_2$ is identical to $R_1$. Connectivity is proven in Lemma \ref{lem:con1}.
\end{proof}

\begin{theorem}
	$R_2$ is deadlock-free.
\end{theorem}
\begin{proof}
	The proof is analog to the proof of Theorem~\ref{thm:r1}. $R_2$ is connected because of Lemma~\ref{lem:con2}. We assume that the CDG of T and $R_2$ has a cycle. Then T has a cycle, in which each two consecutive arcs are direct dependent, according to Lemma~\ref{lem:cyc}.\\
\emph{Case 1: All vertices of the cycle are in the same layer.} Case does not occur, cp.~Theorem~\ref{thm:r1}, \emph{Case 1}.\\
\emph{Case 2: The vertices of the cycle are in at least two different layers.} There is at least one arc going up. According to table \ref{tbl:turns}, the only possible direction after >>up<< is >>up<<. Thus, the cycle can not be closed. Hence, case 2 is impossible.\\
We have shown by contradiction that the CDG is cycle-free. We apply Duato's Theorem on $R_2$.
\end{proof}

\subsection{Livelock freedom}

\emph{Palesi et al.} \cite{Palesi.2014} define that ``livelock is a condition where a packet keeps circulating within the network without ever reaching its destination''. Hence the following definition.
\begin{definition}[Livelock-free] A routing algorithm is livelock-free, if every packet has no other choice, but to reach its destination after a finite number of hops.
\end{definition}
Remark.	A routing algorithm consists of a routing function and a selection. $R_1$ and $R_2$ are examples for routing functions. If an adaptive routing function returns more than one link, the selection chooses one. The property \emph{livelock-free} belongs to the routing algorithm. Nevertheless, we call a routing function \emph{livelock-free} if, independent of the selection, every routing algorithm with this routing function is livelock-free.

\begin{theorem} \label{thm:LLr1}
	$R_1$ and $R_2$ are livelock-free.
\end{theorem}
\begin{proof}
	Assume there were two vertices $s$ and $d$ with the property that the routing $R_1$ makes infinite steps and never reaches $d$ starting from $s$ (the same arguments hold for $R_2$).
	Under this assumption, at least one cardinal direction must be traveled infinite times. We do a case-by-case analysis in which we assume that this applies  to the different cardinal directions. We thereby show that it works for none of them. This contradicts the assumption that there could be a livelock.
	
	\emph{Case 1: \fup\ is traveled infinite times.}
	By the definition of $R_1$ (Definition~\ref{def:R_1}), \tup\ is only used if $v_x = d_x$ and $v_y = d_y$ and $v_z > d_z$, with $v$ being the current vertex. Traveling \tup\ one layer will remain $v_x = d_x$ and $v_y = d_y$ and results either in $v_z = d_z$ or $v_z > d_z$. The only possible direction after \fup\ is \fup. Since there are only $\ell < \infty$ layers, $d$ will be reached after finite steps. Thus, Case 1 can not occur.
	
	\emph{Case 2: \fdown\ is traveled infinite times.}
	Since \tup\ can not be traveled infinite times (Case 1), \tdown\ can not either. It is limited by the layers count, $\ell$, plus the number of times \tup\ could be traveled.
		
	\emph{Case 3: \feast\ and \fwest\ are traveled infinite times.}
		Similar to Case 2, infinite steps to \twest\ imply infinite steps to \teast\ and vice versa. From the definition of $R_1$, we know:
		\begin{itemize}[noitemsep,topsep=0pt, label={--}]
			\item \teast\ and \twest\ are the only directions, which affect the $x$-value of $v$.
			\item A step to \teast\ is only done if $v_x < d_x$
			\item A step to \twest\ is only done if $v_x > d_x$
			\item A step to \twest\ or \teast\ is only done if $v_z \geq d_z$.
		\end{itemize}
		We never step on a router with $v_x = d_x$. If we reached a router with $v_x = d_x$, \tup- or \tdown-routing would be done and the destination would be reached. Steps to \teast\ or \twest\ are only done in the destination layer or below. In these layers, each row has a router at position $d_x$. Routing from west to east and back without using one of these routers is impossible.
		
	\emph{Case 4: \fnorth\ and \fsouth\ are traveled infinite times.}
	This case is analog to Case 3.
	
	None of the cases occur. Thus, the assumption is wrong. $R_1$ is livelock-free.
	
	The same arguments hold for $R_2$ (defined in Definition~\ref{def:R_2}). $R_2$ is livelock-free.
	
	Remark: The proof relies on our special setting. It requires that for $u$ and $v$ with $\down(u) = v$ it holds: $\up(v) = u$, $u_x = v_x$, and $u_y = v_y$. It also requires the mesh topology in layers.
\end{proof}

\section{Tackling throughput via router architectures}\label{sec:architectures}
We have shown a fundamental limitation in heterogeneous routing paths using standard techniques in Sec.~\ref{sec:throughput}: Throughput is limited by the slowest clock along a packet's path, or in other words, \textit{the chain is as strong as its weakest link}. This is not an issue for 2D or homogeneous 3D systems, since the deviation of clocks is rather small there. In heterogeneous 3D SoCs, in contrast, this poses a severe limitation, since clocks potentially deviate by a large factor. This limitation, previously unexplored, is revealed by this paper. To solve this issue in combination with the proposed routing strategies, we propose to use a novel router mircoarchitecture. Thereby, we assume an integer relation between the clock frequencies $c_f$  with a constant phase shift.
Our architecture exploits the observation that optimized routing algorithms must minimize horizontal transmission in slower layers.
With our proposed routing strategies, horizontal transmissions are always conducted in the fastest layer along the path. 
Thus, for heterogeneous packet-paths, packets are directly routed from local ports of a router to the port in direction of the faster layer (down). In the opposite direction, from downwards, packets can only be routed to the upward port or the local port for ejection. The architecture enables a small part of the router in the slower layers, comprising the local and vertical ports, to communicate multiple flits in parallel in order to provide the same throughput between the local and the vertical ports as faster routers from digital layers. Thereby, heterogeneous packet-paths are traversed with the throughput the standard router in the fastest technology provides. We refer to our new architecture as \emph{high vertical-throughput router}.   

\subsection{High vertical-throughput router design}

As previously outlined, the router architecture in the slower layers has to be modified, using parallelism, to obtain a higher throughput between the local and the vertical ports. 
In the fast layers, only the vertical links towards the slower layers need to be modified (see below). Our new router architecture exploits that processing elements, connected to the local ports, are able to provide multiple parallel flits, since packet transmission is initialized for full packets.  A conventional input buffered 3D router design, with link width of $N$, is modified as shown orange, in Fig.~\ref{fig:routerArchitecture}.
\begin{figure}
	\begin{minipage}[b]{0.48\linewidth}
          \centering
          \scalebox{1.0}{
		\begin{tikzpicture}[scale=.3, yscale=-1]

\draw (12, 12) -- ++ (2,0)-- ++ (2,2)-- ++ (0,2) -- ++ (-2,2)-- ++ (-2,0)-- ++ (-2,-2)-- ++ (0,-2)-- ++ (2,-2);
\draw[col2] (12,12.5) -- ++(.5,0) -- ++(1,2) -- ++ (.5,0);
\draw[col2] (12,14.5) -- ++(.5,0) -- ++(1,-2) -- ++ (.5,0);
\node[desc, col2, align=center] (cb) at (13,16) {modified\\crossbar};

\coordinate (north) at  (12.5,9);
\draw[latex-] (north) --++(0,3) node[midway, left, desc, xshift = 2]{$N$};
\draw[-latex] ($(north)+(1,0)$) --++(0,1);
\draw[-latex] ($(north)+(1,0)$) --++(0,3) node[midway, right, desc, xshift=1]{$N$};
\draw[fill=white] ($(north)+(0.5,1)$) rectangle ++(1,.25);
\draw[fill=white] ($(north)+(0.5,1.25)$) rectangle ++(1,.25);
\draw[fill=white] ($(north)+(0.5,1.5)$) rectangle ++(1,.25);
\draw[fill=white] ($(north)+(0.5,1.75)$) rectangle ++(1,.25);
\node[desc, anchor =west] (northdesc) at ($(north)+(1, 0)$) {north};

\coordinate (south) at  (13.5,21);
\begin{scope}[rotate=180]
\draw[latex-] (south) --++(0,3) node[midway, right, desc, xshift = -1]{$N$};
\draw[-latex] ($(south)+(1,0)$) --++(0,1);
\draw[-latex] ($(south)+(1,0)$) --++(0,3) node[midway, left, desc, xshift=-2]{$N$};
\draw[fill=white] ($(south)+(0.5,1)$) rectangle ++(1,.25);
\draw[fill=white] ($(south)+(0.5,1.25)$) rectangle ++(1,.25);
\draw[fill=white] ($(south)+(0.5,1.5)$) rectangle ++(1,.25);
\draw[fill=white] ($(south)+(0.5,1.75)$) rectangle ++(1,.25);
\node[desc, anchor =east] (northdesc) at ($(south)+(1, 0)$) {south};
\end{scope}

\coordinate (south) at  (7,15.5);
\begin{scope}[rotate=-90]
\draw[latex-] (south) --++(0,3) node[midway, below, desc, yshift = 1]{$N$};
\draw[-latex] ($(south)+(1,0)$) --++(0,1);
\draw[-latex] ($(south)+(1,0)$) --++(0,3) node[midway, above, desc, yshift=2]{$N$};
\draw[fill=white] ($(south)+(0.5,1)$) rectangle ++(1,.25);
\draw[fill=white] ($(south)+(0.5,1.25)$) rectangle ++(1,.25);
\draw[fill=white] ($(south)+(0.5,1.5)$) rectangle ++(1,.25);
\draw[fill=white] ($(south)+(0.5,1.75)$) rectangle ++(1,.25);
\node[desc, anchor =west, rotate=90] (northdesc) at ($(south)+(1, 0)$) {west};
\end{scope}

\coordinate (south) at  (19,14.5);
\begin{scope}[rotate=90]
\draw[latex-] (south) --++(0,3) node[midway, above, desc, yshift = -1]{$N$};
\draw[-latex] ($(south)+(1,0)$) --++(0,1);
\draw[-latex] ($(south)+(1,0)$) --++(0,3) node[midway, below, desc, yshift=-2]{$N$};
\draw[fill=white] ($(south)+(0.5,1)$) rectangle ++(1,.25);
\draw[fill=white] ($(south)+(0.5,1.25)$) rectangle ++(1,.25);
\draw[fill=white] ($(south)+(0.5,1.5)$) rectangle ++(1,.25);
\draw[fill=white] ($(south)+(0.5,1.75)$) rectangle ++(1,.25);
\node[desc, anchor =east, rotate=90] (northdesc) at ($(south)+(1, 0)$) {east};
\end{scope}

\coordinate (south) at  (15.5,16.5);
\begin{scope}[rotate=-45, col2, xscale = -1]
\draw[-latex] (south) --++(0,3) node[midway, above, desc, yshift = -3, xshift = -1, rotate =-45]{$c_fN$};
\draw[latex-] ($(south)+(1,2)$) --++(0,1);
\draw[latex-] ($(south)+(1,0)$) --++(0,3) node[midway, below, desc, xshift = 2,yshift=-3, rotate=-45]{$c_fN$};
\draw[fill=white] ($(south)+(0.5,1)$) rectangle ++(1,.25);
\draw[fill=white] ($(south)+(0.5,1.25)$) rectangle ++(1,.25);
\draw[fill=white] ($(south)+(0.5,1.5)$) rectangle ++(1,.25);
\draw[fill=white] ($(south)+(0.5,1.75)$) rectangle ++(1,.25);
\node[desc, anchor =east, rotate=45] (northdesc) at ($(south)+(1, 3.3)$) {local};
\end{scope}

\coordinate (south) at  (14.5,12.5);
\begin{scope}[rotate=45, col2, xscale =1, yscale = -1]
\draw[-latex] (south) --++(0,3) node[midway, above, desc, yshift = -3, xshift = 1, rotate =45]{$c_fN$};
\draw[latex-] ($(south)+(1,2)$) --++(0,1);
\draw[latex-] ($(south)+(1,0)$) --++(0,3) node[midway, below, desc, xshift = 1,yshift=-2, rotate=45]{$c_fN$};
\draw[fill=white] ($(south)+(0.5,1)$) rectangle ++(1,.25);
\draw[fill=white] ($(south)+(0.5,1.25)$) rectangle ++(1,.25);
\draw[fill=white] ($(south)+(0.5,1.5)$) rectangle ++(1,.25);
\draw[fill=white] ($(south)+(0.5,1.75)$) rectangle ++(1,.25);
\node[desc, anchor =west, rotate=-45] (northdesc) at ($(south)+(1, 3.)$) {up};
\end{scope}

\coordinate (south) at  (11.5,17.5);
\begin{scope}[rotate=45, col2, xscale =-1, yscale = 1]
\draw[-latex] (south) --++(0,3) node[midway, above, desc, yshift = 6, xshift = -8, rotate =45]{$c_fN$};
\draw[latex-] ($(south)+(1,0)$) --++(0,1);
\draw[latex-] ($(south)+(1,0)$) --++(0,3) node[midway, below, desc, xshift = 3,yshift=-6, rotate=45]{$c_fN$};
\draw[fill=white] ($(south)+(0.5,1)$) rectangle ++(1,.25);
\draw[fill=white] ($(south)+(0.5,1.25)$) rectangle ++(1,.25);
\draw[fill=white] ($(south)+(0.5,1.5)$) rectangle ++(1,.25);
\draw[fill=white] ($(south)+(0.5,1.75)$) rectangle ++(1,.25);
\node[desc, anchor =east, rotate=-45] (northdesc) at ($(south)+(1, 3.0)$) {down};
\end{scope}

\end{tikzpicture}
	
}
		\caption{Modified router architecture with support for higher vertical throughput. The link width is $N$, and
                  $c_f$ the clock scaling factor of the current layer compared to the fastest layer.}
		\label{fig:routerArchitecture}
	\end{minipage}
	\hfill
	\begin{minipage}[b]{0.48\linewidth}
          \centering
          \scalebox{0.8}{
%
%

\begin{tikzpicture}[scale=.35, yscale=-1]
\coordinate (u1) at (3, 6);
\draw [fill = white]($(u1)+(0,-6)$) rectangle +(2, 8);
\draw ($(u1) + (0,1)$) -- ($(u1) + (2,1)$);
\draw ($(u1) + (0,0)$) -- ($(u1) + (2,0)$);
\draw ($(u1) + (0,-1)$) -- ($(u1) + (2,-1)$);
\draw ($(u1) + (0,-2)$) -- ($(u1) + (2,-2)$);
\draw ($(u1) + (0,-3)$) -- ($(u1) + (2,-3)$);
\draw ($(u1) + (0,-4)$) -- ($(u1) + (2,-4)$);
\draw ($(u1) + (0,-5)$) -- ($(u1) + (2,-5)$);
\node [desc] at ($(u1) + (1, 1.5)$) (b2) {$1$};
\node [desc] at ($(u1) + (1, 0.2)$) (b2) {$\vdots$};
\node [desc] at ($(u1) + (1, -0.5)$) (b2) {$c_f$};
\node [desc] at ($(u1) + (1, -3.5)$) (b2) {$1$};
\node [desc] at ($(u1) + (1, -4.8)$) (b2) {$\vdots$};
\node [desc] at ($(u1) + (1, -5.5)$) (b2) {$c_f$};
\draw[latex-] ($(u1) + (0,-5.5)$) -- ($(u1) + (-1.5,-5.5)$);
\node [desc] at ($(u1) + (-1, -4.8)$) (b2) {$\vdots$};
\draw[latex-] ($(u1) + (0,-3.5)$) -- ($(u1) + (-1.5,-3.5)$);
\node[circle, fill, inner sep=1.1pt] at($(u1)+(-1.5,-5.5)$) (b31) {};
\node[circle, fill, inner sep=1.1pt] at($(u1)+(-1.5,-3.5)$) (b31) {};
\draw[thick] ($(u1) + (-1.5,-7.5)$) -- ($(u1) + (-1.5,-3.5)$);
\node at ($(u1) + (-1.5,-8.25)$) {\small In};
\node at ($(u1) + (-2.5,-6.5)$) {\small $c_fN$};
\node at ($(u1) + (-0.75,-6)$) {\small $N$};
\draw[latex-] ($(u1) + (3.5,4.5)$) |- ($(u1) + (2,1.5)$);
\node [desc] at ($(u1) + (2.75, 0.2)$) (b2) {$\vdots$};
\draw[latex-] ($(u1) + (6,-0.5)$) -- ($(u1) + (2,-0.5)$);
\draw[latex-,thick] ($(u1) + (6,4.5)$) -- ($(u1) + (6,-0.5)$);
\draw[latex-] ($(u1) + (6,1.5)$) -- ($(u1) + (2,1.5)$);
\node[circle, fill, inner sep=1.1pt] at($(u1)+(6,-0.45)$) (b31) {};
\node[circle, fill, inner sep=1.1pt] at($(u1)+(6,1.55)$) (b31) {};
\node at ($(u1) + (6.15,5.25)$) {\small Out$_P$};
\node at ($(u1) + (3.5,5.25)$) {\small Out$_S$};
\node at ($(u1) + (7.25,2.5)$) {\small $c_fN$};
\node at ($(u1) + (3.5,-1)$) {\small $N$};
\node at ($(u1) + (3,2.5)$) {\small $N$};


\end{tikzpicture}
	
}
		\caption{Modified input buffer. $c_f$ flits can be read and written at once.}
		\label{fig:mod_input_buffs}
                \vspace{10pt}
	\end{minipage}
\end{figure}
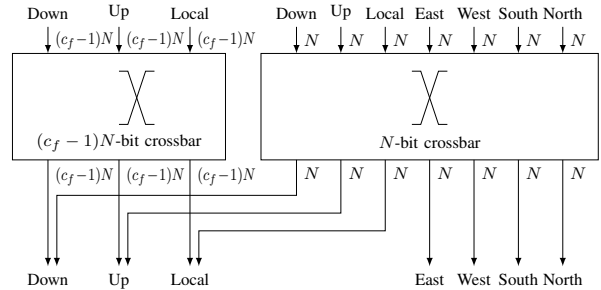
\begin{figure}
	\centering
	                \scalebox{0.8}{
                  \newcommand{\inputArrow}[4]{\coordinate (tmpArrowLower) at (#1,#2);
	\coordinate (tmpArrowUpper) at ($(tmpArrowLower)+(0,-1.5)$);
	\draw[latex-] (tmpArrowLower) --  (tmpArrowUpper)  node [midway, right, desc] {#4};
	\node[above = 0pt of tmpArrowUpper, desc] {#3};}
\newcommand{\outputArrow}[4]{\coordinate (tmpArrowUpper) at (#1,#2);
	\coordinate (tmpArrowLower) at ($(tmpArrowUpper)+(0,6)$);
	\draw[-latex] (tmpArrowUpper) node [below right, desc] {#4}--  (tmpArrowLower);
	\node[below = 0pt of tmpArrowLower, desc] {#3};}

\begin{tikzpicture}[scale=.295, yscale=-1]

\coordinate (rec2ul) at (11+1, 6);
\coordinate (rec2lr) at (30+1,12);
\draw [fill = white] (rec2ul) rectangle (rec2lr);
\draw (19.5+1, 7) -- (20+1, 7) -- (21+1, 10) -- (21.5+1, 10);
\draw (21.5+1, 7) -- (21+1, 7) -- (20+1, 10) -- (19.5+1, 10);
\node [desc, align=center]  at (12/2+31/2, 11) (c1) {$N$-bit crossbar};

\coordinate (rec1ul) at (-2, 6);
\coordinate (rec1lr) at (10,12);
\draw [fill = white] (rec1ul) rectangle (rec1lr);
\draw (4, 7) -- (4.5, 7) -- (5.5, 10) -- (6, 10);
\draw (6, 7) -- (5.5, 7) -- (4.5, 10) -- (4, 10);
\node [desc]  at (-1+5 , 11) (c1) {$(c_f-1)N$-bit crossbar};


\inputArrow{0}{6} {Down}{\scalebox{.8}[1.0]{$(c_f\!-\!1)N$}}
\inputArrow{4}{6} {Up}{\scalebox{.8}[1.0]{$(c_f\!-\!1)N$}}
\inputArrow{8}{6} {Local}{\scalebox{.8}[1.0]{$(c_f\!-\!1)N$}}

\inputArrow{13+1}{6} {Down}{$N$}
\inputArrow{15.5+1}{6} {Up}{$N$}
\inputArrow{18+1}{6} {Local}{$N$}
\inputArrow{20.5+1}{6} {East}{$N$}
\inputArrow{23+1}{6} {West}{$N$}
\inputArrow{25.5+1}{6} {South}{$N$}
\inputArrow{28+1}{6} {North}{$N$}

\outputArrow{-0}{12}{Down}{\scalebox{.8}[1.0]{$(c_f\!-\!1)N$}}
\outputArrow{4}{12}{Up}{\scalebox{.8}[1.0]{$(c_f\!-\!1)N$}}
\outputArrow{8}{12}{Local}{\scalebox{.8}[1.0]{$(c_f\!-\!1)N$}}
\outputArrow{20.5+1}{12}{East}{$N$}
\outputArrow{23+1}{12}{West}{$N$}
\outputArrow{25.5+1}{12}{South}{$N$}
\outputArrow{28+1}{12}{North}{$N$}

\draw[-latex]  (13+1, 12) node [below right, desc] {$N$} -- ++(0,2) -- ++(-13.5,0) -- ++(0,4);
\draw[-latex]  (15.5+1, 12) node [below right, desc] {$N$} -- ++(0,3) -- ++(-12,0) -- ++(0,3);
\draw[-latex]  (18+1, 12) node [below right, desc] {$N$} -- ++(0,4) -- ++(-10.5,0) -- ++(0,2);

%
%

\end{tikzpicture}
	
}
                \caption{Modified crossbar which allows to route $c_f$ flits between the local and the vertical ports.}
		\label{fig:mod_crossbar}

\end{figure}
The input-buffers (see Fig.~\ref{fig:mod_input_buffs}) of the vertical and local ports can read up to $c_f$ flits of $N$ bits simultaneously. A single or $c_f$ flits are inputed to the crossbar, which increases the bit width of the connection by factor $c_f$.
The crossbar is also modified (see Fig.~\ref{fig:mod_crossbar}). Firstly, due to the proposed routing strategies, some turns (e.g.\ down to north, east, west or south) cannot occur. Secondly, the crossbar has to be extended to route $c_f$ flits between local and vertical ports. In paths which do not include the fastest layer, horizontal routes via a slower layer cannot be avoided (still the fastest among all included ones is chosen). In this scenario, routes of single flits from the horizontal ports towards the up or local output ports occur. All remaining $(c_f-1)N$ lines of the crossbar output are zero and only one flit can be written to the local port, or the input port of the overlying router, per cycle.

However, in the most common 3D NoC scenario with only one slower (mixed-signal) layer located at the top, the complexity of the proposed router architecture is reduced drastically for two reasons.
Firstly, the modified routers at the top have no up port. This results in only tree ports, local, up and down, requiring a high-throughput connection. Thereby, the $(c_f-1)N$-bit crossbar  shown in Fig.~\ref{fig:mod_crossbar} is added to the design; it has only three input and output ports. Thus, the local input port is directly connected to the downwards output port and vice versa, which does not incur any hardware cost. Furthermore, only two input buffers (local and down) need to be modified, which again reduces the hardware complexity. 
Secondly, all heterogeneous packets path will include a fast layer, thus routes of single flits from/to the downwards input ports will not occur. This again reduces the complexity of the input buffer as it only needs to send and/or receive $c_f$ parallel flits and never single flits.

\subsection{High vertical-throughput links}

The vertical links must support the higher throughput of the modified routers.
$c_f$ flits are transmitted in parallel employing a large MIV array. (A large TSV array can also be implemented in case of if non-monolithic 3D integration.) 
On the way from a slower to a faster layer, data is transmitted in parallel with the slower clock frequency via the MIV array. The modified input buffer in the faster technology fetches the $c_f$ flits in parallel with a rate equal to the clock speed of the slower layer. If data are transmitted to a slower layer, the data is first parallelized  in the faster layer using a shift register. The full content of the $c_fN$-bit shift register is transmitted via the wide MIV array to the slower layer, where the flits are fetched in parallel by the modified input buffer. This is shown in Fig.~\ref{fig:commUpMany}. The inverse path from the slower layer to the faster layer is shown in Fig.~\ref{fig:commDownMIV}. The architecture is analogous; Flits are transmitted in parallel from the slower layer and serialized using a shift register in the faster layer.

\begin{figure}
	\begin{minipage}[b]{.45\linewidth}
          \centering
          \scalebox{0.8}{
%
%

\begin{tikzpicture}[scale=.35, yscale=-1]


\coordinate (u1) at (2, -1);
\draw [fill = white](u1) rectangle +(2, 4);
\draw ($(u1) + (0,1)$) -- ($(u1) + (2,1)$);
\draw ($(u1) + (0,2)$) -- ($(u1) + (2,2)$);
\node [desc] at ($(u1) + (1, 1.5)$) (b2) {$1$};
\node [desc] at ($(u1) + (1, 2.25)$) (b2) {$\vdots$};
\node [desc] at ($(u1) + (1, 3.55)$) (b2) {$c_f$};
\draw ($(u1) + (0,3)$) -- ($(u1) + (2,3)$);
\coordinate (t1) at (3, 4);
\draw [-latex] ($(t1) + (0, -.1)$) -- ($(t1) + (0, -.9)$);
\draw ($(t1) + (0,0.3)$) circle (.15 and 0.1);
\draw (t1) -- +(0,0.3);
\draw ($(t1) + (-.15, .3)$) -- + (0,.5);
\draw ($(t1) + (.15, .3)$) -- + (0,.5);
\begin{scope}
\clip ($(t1) + (-.2, .8)$) rectangle ($(t1) + (.2, 9)$);
\draw ($(t1) + (0,0.8)$) circle (.15 and 0.1);
\path ($(t1) + (-.15, .8)$) -- ($(t1) + (.15, .8)$);
\end{scope}
\draw ($(t1) + (0,0.9)$) -- +(0,0.2);
\draw ($(t1) + (-.35, .37)$) -- + (.7,.4);
\node [desc, anchor = west] at ($(t1) + (.15, .6)$) (td1) {$N$};
\draw [latex-] ($(t1) + (0, 1.2)$) -- ($(t1) + (0, 1.9)$);
\begin{scope}[xshift=-90]
\coordinate (t1) at (3, 4);
\draw [-latex] ($(t1) + (0, -.1)$) |- ($(t1) + (2.25, -3.5)$);
\draw ($(t1) + (0,0.3)$) circle (.15 and 0.1);
\draw (t1) -- +(0,0.3);
\draw ($(t1) + (-.15, .3)$) -- + (0,.5);
\draw ($(t1) + (.15, .3)$) -- + (0,.5);
\begin{scope}
\clip ($(t1) + (-.2, .8)$) rectangle ($(t1) + (.2, 9)$);
\draw ($(t1) + (0,0.8)$) circle (.15 and 0.1);
\path ($(t1) + (-.15, .8)$) -- ($(t1) + (.15, .8)$);
\end{scope}
\draw ($(t1) + (0,0.9)$) -- +(0,0.2);
\draw ($(t1) + (-.35, .37)$) -- + (.7,.4);
\node [desc, anchor = west] at ($(t1) + (.15, .6)$) (td1) {$N$};
\draw [latex-] ($(t1) + (0, 1.2)$) -- ($(t1) + (0, 1.9)$);
\end{scope}
\node [desc, anchor = west] at ($(t1) + (1.15, .6)$) (td1) {$\dots$};
\coordinate (u1) at (-0.75, 6);
\draw [fill = lightgray](u1) rectangle +(4.25, 1);
\draw [fill = lightgray](u1) rectangle +(3, 1);
\draw [fill = lightgray](u1) rectangle +(1.25, 1);
\node[desc] at ($(u1)+(0.6,0.5)$) {$1$};
\node[desc] at ($(u1)+(2.3,0.5)$) {$\dots$};
\node[desc] at ($(u1)+(3.6,0.5)$) {$c_f$};
\node [desc,gray] at ($(u1) + (2.5, 1.75)$) (b2) {$c_fN$-bit shift-in reg.};
\draw [latex-] ($(u1) + (4.25, 0.5)$) -| ($(u1) + (8.25, 2.5)$);
\node [desc] at ($(u1) + (8.25, 3.25)$) (b2) {router out};
\node [desc] at ($(u1) + (8.25, 4.25)$) (b2) {Up};
\node [desc] at ($(u1) + (8.75, 1.25)$) (b2) {$N$};
\draw [col3, dashed] (4.7,4.5) -- (10.5, 4.5);
\draw [col3, dashed] (-2,4.5) -- (-0.75, 4.5);
\node [rotate = 90, col3, desc] at (10, 7.5) (d1) {fast layer};
\node [rotate = 90, col3, desc] at (10, 2) (d1) {slow layer};

\end{tikzpicture}
	
}
		\vspace{-11.75pt}
		\caption{High-throughput connection from a faster layer to a slower layer employing a large MIV array and a shift register.}
		\label{fig:commUpMany}
              \end{minipage}\hfill
          	\begin{minipage}[b]{.45\linewidth}
          	\centering
          	\scalebox{0.8}{
%
%

\begin{tikzpicture}[scale=.35, yscale=-1]


\coordinate (u1) at (2, -1);
\draw [fill = white](u1) rectangle +(2, 4);
\draw ($(u1) + (0,1)$) -- ($(u1) + (2,1)$);
\draw ($(u1) + (0,2)$) -- ($(u1) + (2,2)$);
\node [desc] at ($(u1) + (1, 1.5)$) (b2) {$1$};
\node [desc] at ($(u1) + (1, 2.25)$) (b2) {$\vdots$};
\node [desc] at ($(u1) + (1, 3.55)$) (b2) {$c_f$};
\draw ($(u1) + (0,3)$) -- ($(u1) + (2,3)$);
\coordinate (t1) at (3, 4);
\draw [latex-] ($(t1) + (0, -.1)$) -- ($(t1) + (0, -.9)$);
\draw ($(t1) + (0,0.3)$) circle (.15 and 0.1);
\draw (t1) -- +(0,0.3);
\draw ($(t1) + (-.15, .3)$) -- + (0,.5);
\draw ($(t1) + (.15, .3)$) -- + (0,.5);
\begin{scope}
\clip ($(t1) + (-.2, .8)$) rectangle ($(t1) + (.2, 9)$);
\draw ($(t1) + (0,0.8)$) circle (.15 and 0.1);
\path ($(t1) + (-.15, .8)$) -- ($(t1) + (.15, .8)$);
\end{scope}
\draw ($(t1) + (0,0.9)$) -- +(0,0.2);
\draw ($(t1) + (-.35, .37)$) -- + (.7,.4);
\node [desc, anchor = west] at ($(t1) + (.15, .6)$) (td1) {$N$};
\draw [-latex] ($(t1) + (0, 1.2)$) -- ($(t1) + (0, 1.9)$);
\begin{scope}[xshift=-90]
\coordinate (t1) at (3, 4);
\draw [latex-] ($(t1) + (0, -.1)$) |- ($(t1) + (2.25, -3.5)$);
\draw ($(t1) + (0,0.3)$) circle (.15 and 0.1);
\draw (t1) -- +(0,0.3);
\draw ($(t1) + (-.15, .3)$) -- + (0,.5);
\draw ($(t1) + (.15, .3)$) -- + (0,.5);
\begin{scope}
\clip ($(t1) + (-.2, .8)$) rectangle ($(t1) + (.2, 9)$);
\draw ($(t1) + (0,0.8)$) circle (.15 and 0.1);
\path ($(t1) + (-.15, .8)$) -- ($(t1) + (.15, .8)$);
\end{scope}
\draw ($(t1) + (0,0.9)$) -- +(0,0.2);
\draw ($(t1) + (-.35, .37)$) -- + (.7,.4);
\node [desc, anchor = west] at ($(t1) + (.15, .6)$) (td1) {$N$};
\draw [-latex] ($(t1) + (0, 1.2)$) -- ($(t1) + (0, 1.9)$);
\end{scope}
\node [desc, anchor = west] at ($(t1) + (1.15, .6)$) (td1) {$\dots$};
\coordinate (u1) at (-0.75, 6);
\draw [fill = lightgray](u1) rectangle +(4.25, 1);
\draw [fill = lightgray](u1) rectangle +(3, 1);
\draw [fill = lightgray](u1) rectangle +(1.25, 1);
\node[desc] at ($(u1)+(0.6,0.5)$) {$1$};
\node[desc] at ($(u1)+(2.3,0.5)$) {$\dots$};
\node[desc] at ($(u1)+(3.6,0.5)$) {$c_f$};
\node [desc,gray] at ($(u1) + (2.5, 1.75)$) (b2) {$c_fN$-bit shift-in reg.};
\draw [-latex] ($(u1) + (4.25, 0.5)$) -| ($(u1) + (8.25, 2.5)$);
\node [desc] at ($(u1) + (8.25, 3.25)$) (b2) {router in};
\node [desc] at ($(u1) + (8.25, 4.25)$) (b2) {Up};
\node [desc] at ($(u1) + (8.75, 1.25)$) (b2) {$N$};
\draw [col3, dashed] (4.7,4.5) -- (10.5, 4.5);
\draw [col3, dashed] (-2,4.5) -- (-0.75, 4.5);
\node [rotate = 90, col3, desc] at (10, 7.5) (d1) {fast layer};
\node [rotate = 90, col3, desc] at (10, 2) (d1) {slow layer};

\end{tikzpicture}
	
}
          	\vspace{-11.75pt}
          	\caption{High-throughput connection from a slower layer to a faster layer employing a large MIV array and a shift register.}
          	\label{fig:commDownMIV}
          \end{minipage}
\end{figure}

%

\section{Results}\label{sec:results}

This section consists of four parts: First, we discuss the accuracy of our models for a set of commercial mixed-signal and digital technology nodes in Sec.~\ref{sec:results:modelfit}. Second, we show the impact of latency of our routing algorithms for \unit[130]{nm} commercial mixed-signal and  \unit[90]{nm} -- \unit[28]{nm} commercial digital nodes in Sec.~\ref{sec:results:latency}. Third, we focus on our router architectures by analyzing throughput improvements in Sec.~\ref{sec:results:throughput}. Forth, we conclude the co-design of routing strategies and algorithms by considering the implementation costs and power improvements in Sec.~\ref{sec:results:area}. Finally, we show the practical applicability of our proposed solution for heterogeneous 3D interconnects by means of a 3D VSoC case study in Sec.~\ref{sec:results:casestudy} using a heterogeneous combination of \unit[30]{nm} mixed-signal and \unit[15]{nm} digital technology nodes.

\subsection{Model accuracy}\label{sec:results:modelfit}

First, we present results on the model accuracy of our area and timing model. Second, we give simulation results that support our claim of accurately modeling communication under zero load.

We fit the area and timing model to the synthesis results of a 3D NoC router with two virtual channels, four flit deep buffers per channel, credit based flow control, wormhole switching, decentralized arbiters and deterministic XYZ-routing using Synopsys design compiler for commercial \unit[130]{nm} mixed-signal technology and commercial \unit[28 -- 90]{nm} digital technology. We use both general purpose (GP) and ultra low voltage (ULV) mixed-signal technology to exemplify potential differences. The synthesis results are used to evaluate the accuracy of the model fit.

The synthesis results and the fitted models for the \emph{area scaling factor} are shown in Fig.~\ref{fig:nocrouterScalingArea}. Curve fitting is conducted with Mathematica 10. 
The example yields a non-ideality factor $\alpha = 3462.7$ and an offset of $\hat{\alpha} = 29.8$ for \unit[130]{nm} GP technology with a root mean square error (RMSE) of 0.1286. ULV technology yields $\alpha = 13.2$ and an offset of $\hat{\alpha} = 0.124$ with a RMSE of 0.1414. 

The synthesis results and the fitted models for the \emph{clock scaling factor} with a predicted maximum achievable clock frequency of \unit[5.0]{GHz} are shown in Fig.~\ref{fig:nocrouterScalingTiming}. (Smaller commercial technology nodes below \unit[28]{nm} are not available, thus we set $\beta$ instead of fitting it to the model.) The fitting is conducted with Mathematica 10. The results for GP nodes are $\beta =  32.85$, $\hat{\beta} = 7.88$, $\tilde{\beta} = 0.76$, and $\bar{\beta} = 1.26$ with a RMSE of 0.30. For ULV nodes, the model yields the parameters $\beta =  77.45$, $\hat{\beta} = 2.48$, $\tilde{\beta} = 0.76$ and $\bar{\beta} = 2.77$, with a RMSE of 0.71.

We claimed that our models for head flit latency are accurate under zero load by construction. To validate this, we use simulations for the latency enhancement of packets traversing the network for our two proposed routing strategies. These are shown in Figs.~\ref{fig:SpeedupAlgo1} and \ref{fig:SpeedupAlgo2}. We report the latency enhancement both from model and simulations. One can see that the results are matching and that our models, indeed, are accurate under zero load.

\begin{figure}
		\centering
		\includegraphics[width=\linewidth]{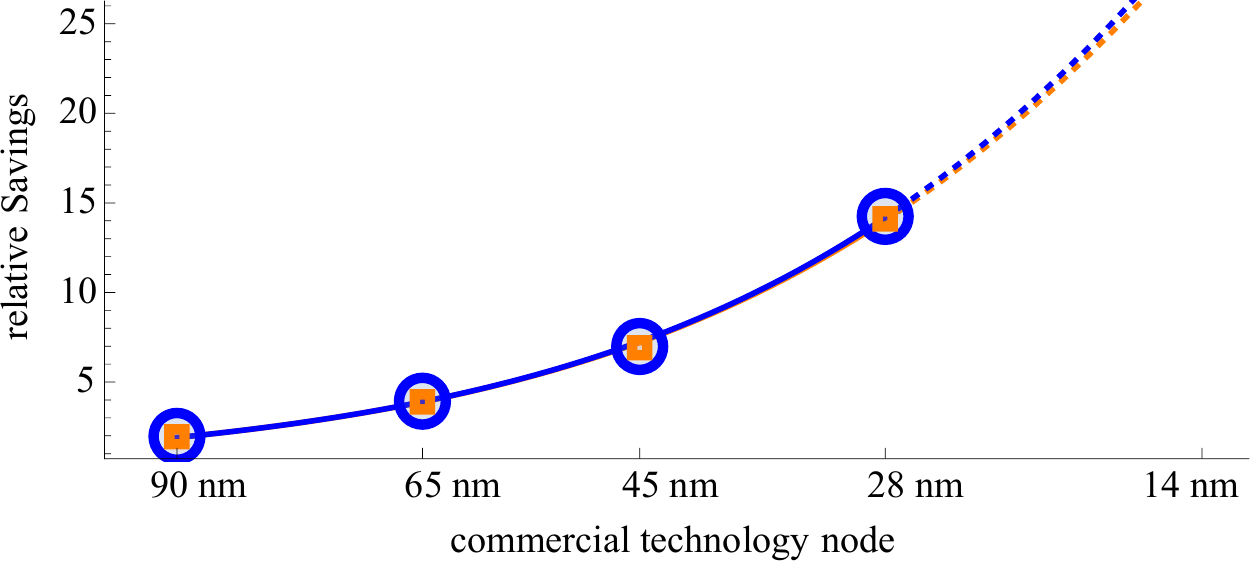}
		\caption{Area model accuracy using exemplary fit (orange -- ULV, blue -- GP).}
		\label{fig:nocrouterScalingArea}
\end{figure}
\begin{figure}
		\centering
		\includegraphics[width=\linewidth]{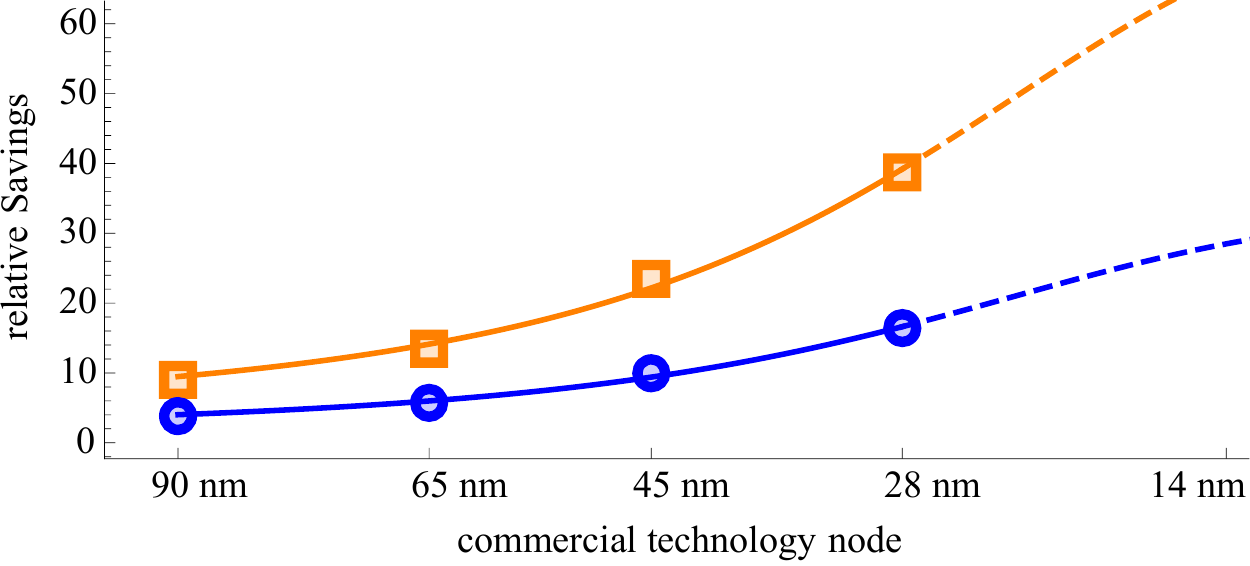}
		\caption{Timing model accuracy using exemplary fit (predictive maximum achievable clock frequency of  \unit[5]{GHz}; orange -- ULV, blue -- GP).}
		\label{fig:nocrouterScalingTiming}
\end{figure}

\subsection{Latency of routing algorithms}\label{sec:results:latency}

\subsubsection{Latency of \AlgorithmI}

Packets from any node in the mixed-signal layers to any node in the digital layers profit from \AlgorithmI. We compare their latency under zero load to \XYZ. As an exemplary use case, we use a 3D SoC, which consists of two layers: One in a commercial mixed-signal technology implementing a 4x4 NoC and one in \unit[90]{nm} -- \unit[28]{nm} commercial digital node implementing a NoC with more nodes according to the area model (Eq.~\ref{eq:area}) on basis of synthesis results. The achieved speedup is calculated using both a cycle-accurate NoC simulator with 16 flit deep buffer, wormhole routers and four VCs \cite{Joseph.2016c} and $\Delta_H$ from Eq.~\ref{eq:packetTransmissionTime}. The results are shown in Fig.~\ref{fig:SpeedupAlgo1} for all available hop distances in the layer in mixed-signal technology. Simulation and model results are identical; the model is accurate under zero load. The latency speedup is between $1.5\times$ and $6.5\times$. It is larger if a more advanced digital node is used, which is consistent with the expectations from Sec.~\ref{sec:potentials}. Note that this speedup is achieved without any implementation costs.

\begin{figure*}
	\centering
%
\pgfplotstableread[col sep = comma]{tables/ZupXYZdownSpeedUp.csv}\SpeedupData
\begin{tikzpicture}
	\begin{axis}[scale = 0.75,
	ybar,
	bar width=0.075cm,
	grid = both,
	ymin = 0,
	xtick={0.25, .5, .75, 1, 1.25, 1.5},
	xticklabels={1, 2, 3, 4, 5, 6},
	xlabel=hop distance in 4x4 mixed signal layer, 
	legend pos=outer north east ,
	height = 4.48cm,
	width=14cm,
	ylabel style={align=center,desc, yshift = -20pt},
	xlabel style={desc, yshift = 2pt, align=center},
	legend style={desc},
	xticklabel style={desc, align = center},
	yticklabel style={desc},
	ylabel=latency enhancement
	]
	
	\addlegendimage{empty legend}
	\addlegendimage{empty legend}
	\addlegendimage{empty legend}
	\addlegendimage{empty legend}
	\addlegendimage{empty legend}

	\addplot [col2, fill=col2] table[x=s,y=speedup130] {\SpeedupData};
	\addplot [col2, fill=col2!10] table[x=s,y=speedup130] {\SpeedupData};

	\addplot [col1, fill=col1] table[x=s,y=speedup90] {\SpeedupData};
	\addplot [col1, fill=col1!10] table[x=s,y=speedup90] {\SpeedupData};
	
	\addplot [col3, fill=col3] table[x=s,y=speedup65] {\SpeedupData};
	\addplot [col3, fill=col3!10] table[x=s,y=speedup65] {\SpeedupData};

	\addplot [red, fill=red] table[x=s,y=speedup45] {\SpeedupData};
	\addplot [red, fill=red!10] table[x=s,y=speedup45] {\SpeedupData};

	\addlegendentry{\hspace{-.4cm}\bfseries\unit[130]{nm} mixed signal and}

	\addlegendentry{\unit[90]{nm} dig.: sim.\draw[col2, fill = col2, /tikz/.cd,yshift=-0.25em](0cm,0cm) rectangle (3pt,0.8em);\enspace, model\! \draw[col2, fill = col2!10, /tikz/.cd,yshift=-0.25em](0cm,0cm) rectangle (3pt,0.8em);}
	\addlegendentry{\unit[65]{nm} dig.: sim.\draw[col1, fill = col1, /tikz/.cd,yshift=-0.25em](0cm,0cm) rectangle (3pt,0.8em);\enspace, model\! \draw[col1, fill = col1!10, /tikz/.cd,yshift=-0.25em](0cm,0cm) rectangle (3pt,0.8em);}
	\addlegendentry{\unit[45]{nm} dig.: sim.\draw[col3, fill = col3, /tikz/.cd,yshift=-0.25em](0cm,0cm) rectangle (3pt,0.8em);\enspace, model\! \draw[col3, fill = col3!10, /tikz/.cd,yshift=-0.25em](0cm,0cm) rectangle (3pt,0.8em);}
	\addlegendentry{\unit[28]{nm} dig.: sim.\draw[red, fill = red, /tikz/.cd,yshift=-0.25em](0cm,0cm) rectangle (3pt,0.8em);\enspace, model\! \draw[red, fill = red!10, /tikz/.cd,yshift=-0.25em](0cm,0cm) rectangle (3pt,0.8em);}
\end{axis}
\end{tikzpicture}
	
	\vspace{-14pt}
	\caption{Latency enhancement of \AlgorithmI\ to \XYZ.}
	\label{fig:SpeedupAlgo1}
\end{figure*}
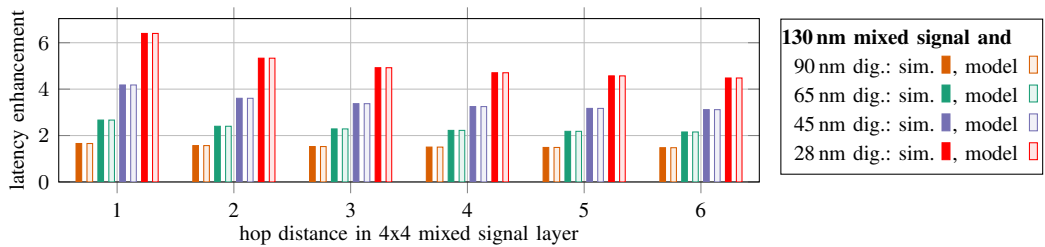

\subsubsection{Latency of \AlgorithmII}
Packets from any node in the mixed-signal layers to any node in the mixed-signal layers profit from \AlgorithmII. Again, we compare their latency under zero load to \XYZ. As an exemplary use case, we use a the same 3D SoC as before with two layers. The achieved speedup is calculated using both a cycle-accurate NoC simulator with 16 flit deep buffer, wormhole routers and four VCs and $\Delta_H$ and $\Delta_V$ from Eqs.~\ref{eq:packetTransmissionTime}, \ref{eq:latencyVdown} and \ref{eq:latencyVup}. The results are shown in Fig.~\ref{fig:SpeedupAlgo2} for all available hop distances in the layer in mixed-signal technology. The latency speedup is between $0.54\times$ and $1.79\times$. It is noteworthy that any speedup is achieved with negligible implementation costs.

\begin{figure*}
	\centering
%
\pgfplotstableread[col sep = comma]{tables/ZXYZSpeedUp.csv}\SpeedupDataZYXZ
\begin{tikzpicture}
	\begin{axis}[scale = 0.75,
	ybar,
	bar width=0.085cm,
	grid = both,
	ymin = -0.65,
	ymax = 1.15,
	xtick={0.25, .5, .75, 1, 1.25, 1.5},
	xticklabels={1, 2, 3, 4, 5, 6},
	ytick = {-.5,0,.5, 1},
	yticklabels={0.5, 1, 1.5, 2},  
	xlabel=hop distance in 4x4 layer in mixed signal node, 
	legend pos=outer north east ,
	height = 4.48cm,
	width=14.1cm,
	ylabel style={align=center,desc, yshift = -10pt},
	xlabel style={desc, yshift = 2pt, align=center},
	legend style={desc},
	xticklabel style={desc, align = center},
	yticklabel style={desc},
	ylabel=latency enhancement
	]
	
	\addlegendimage{empty legend}
	\addlegendimage{empty legend}
	\addlegendimage{empty legend}
	\addlegendimage{empty legend}
	\addlegendimage{empty legend}

	\addplot [col2, fill=col2] table[x=s,y=speedup130corr] {\SpeedupDataZYXZ};
	\addplot [col2, fill=col2!10] table[x=s,y=speedup130corr] {\SpeedupDataZYXZ};
	
	\addplot [col1,fill=col1] table[x=s,y=speedup90corr] {\SpeedupDataZYXZ};
	\addplot [col1, fill=col1!10] table[x=s,y=speedup90corr] {\SpeedupDataZYXZ};
		
	\addplot [col3, fill=col3] table[x=s,y=speedup65corr] {\SpeedupDataZYXZ};
	\addplot [col3, fill=col3!10] table[x=s,y=speedup65corr] {\SpeedupDataZYXZ};
	
	\addplot [red, fill=red] table[x=s,y=speedup45corr] {\SpeedupDataZYXZ};
	\addplot [red, fill=red!10] table[x=s,y=speedup45corr] {\SpeedupDataZYXZ};
	 
	\addlegendentry{\hspace{-.4cm}\bfseries \unit[130]{nm} mixed signal and}

	\addlegendentry{\unit[90]{nm} dig.: sim.\draw[col2, fill = col2, /tikz/.cd,yshift=-0.25em](0cm,0cm) rectangle (3pt,0.8em);\enspace, model\! \draw[col2, fill = col2!10, /tikz/.cd,yshift=-0.25em](0cm,0cm) rectangle (3pt,0.8em);}
	\addlegendentry{\unit[65]{nm} dig.: sim.\draw[col1, fill = col1, /tikz/.cd,yshift=-0.25em](0cm,0cm) rectangle (3pt,0.8em);\enspace, model\! \draw[col1, fill = col1!10, /tikz/.cd,yshift=-0.25em](0cm,0cm) rectangle (3pt,0.8em);}
	\addlegendentry{\unit[45]{nm} dig.: sim.\draw[col3, fill = col3, /tikz/.cd,yshift=-0.25em](0cm,0cm) rectangle (3pt,0.8em);\enspace, model\! \draw[col3, fill = col3!10, /tikz/.cd,yshift=-0.25em](0cm,0cm) rectangle (3pt,0.8em);}
	\addlegendentry{\unit[28]{nm} dig.: sim.\draw[red, fill = red, /tikz/.cd,yshift=-0.25em](0cm,0cm) rectangle (3pt,0.8em);\enspace, model\! \draw[red, fill = red!10, /tikz/.cd,yshift=-0.25em](0cm,0cm) rectangle (3pt,0.8em);}
\end{axis}
\end{tikzpicture}
	
	\vspace{-14pt}
	\caption{Latency enhancement of \AlgorithmII\ to \XYZ.}
	\label{fig:SpeedupAlgo2}
\end{figure*}
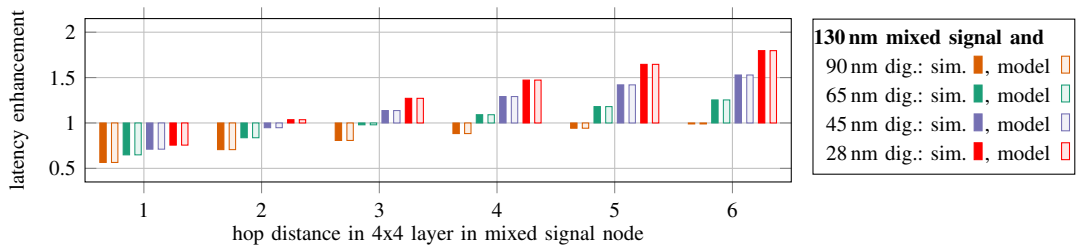

\subsection{Throughput of high vertical-throughput router}\label{sec:results:throughput}

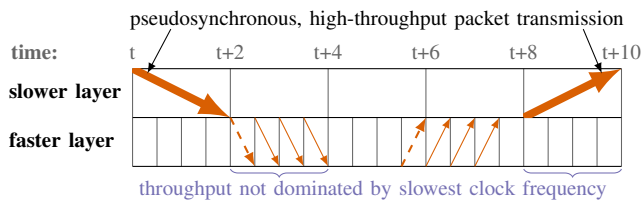
\begin{figure}
	\centering

\newcommand{\drawpacketClock}[6]{
\draw[-latex, #4, thick, dashed, #3] (#1,#2) -- (#1+#5, #2+1);
\foreach \i  [evaluate=\i as \j using \i*#6]  in {1,..., 3}{
\draw[-latex, #4, #3] (#1+\j,#2) -- (#1+\j+#5, #2+1);
}
}

\newcommand{\drawpacketClockUp}[6]{
	\draw[-latex, #4, thick, dashed, #3] (#1,#2) -- (#1+#5, #2-1);
	\foreach \i  [evaluate=\i as \j using \i*#6]  in {1,..., 3}{
		\draw[-latex, #4, #3] (#1+\j,#2) -- (#1+\j+#5, #2-1);
	}
}


	\begin{tikzpicture}[scale = 0.65, yscale=-1,
cube/.style={thick, black, col1, fill = col3, fill opacity = 0.2},
axis/.style={-latex,col2, thick}, digital/.style={thick, col1}, digitalLines/.style = {col1}]

\foreach \i in {0,2,...,10}{
	\draw[black!60, thin] (\i, 0) -- (\i, 1);
}
\foreach \i in {0,0.5,...,10}{
	\draw[black!60, thin] (\i, 1) -- (\i, 2);
}
\foreach \i in {0,...,2}{
	\draw[] (0,\i) -- (10,\i);
}
\node[desc, anchor = west, align = left, xshift = -4pt] at (-2.5, .5)   (a) {\textbf{slower layer}};
\node[desc, anchor = west, align = left, xshift = -4pt] at (-2.5, 1.5)   (a) {\textbf{faster layer}};

\foreach \i in {2,4, ...,10}{
	\node[yshift = -0pt,desc, anchor = south, align = center, black!60] at (\i, 0)   (a) {t+\i};
}
\node[yshift = -0pt,desc, anchor = south, align = center, black!60] at (0, 0)   (a) {t};
\node[yshift = -0pt,desc, anchor = south, align = center, black!60] at (-2, 0)   (a) {\textbf{time:}};

\begin{scope}
\clip (0, 0) rectangle (2, 1);
\draw[-latex, line width = 3pt, col2] (0, 0) -- (2, 1);
\end{scope}
\drawpacketClock{2}{1}{}{col2}{.5}{.5}

\begin{scope}
\clip (8, 1) rectangle (10, 0);
\draw[-latex, line width = 3pt, col2] (8, 1) -- (10,0);
\end{scope}
\drawpacketClockUp{5.5}{2}{}{col2}{.5}{0.5}


\draw [decorate,decoration={mirror,brace,amplitude=3pt	,raise=1pt},yshift=0pt, draw = col3]
(2,2) -- (4,2);
\draw [decorate,decoration={mirror,brace,amplitude=3pt	,raise=1pt},yshift=0pt, draw = col3]
(8,2) -- (10,2);
\node [desc,col3, xshift = -20pt,yshift = -2pt, anchor = north, align=center] at(6 , 2) {throughput not dominated by slowest clock frequency};


\node [desc, align = center, anchor = south] at (5, -0.6) (superpacket) {pseudosynchronous, high-throughput packet transmission};
\draw [-latex] ([yshift=-6pt, xshift =-120pt]superpacket.south) -- (0.3, 0.1);
\draw [-latex] ([yshift=-6pt, xshift =120pt]superpacket.south) -- (9.6, 0.1);
\end{tikzpicture}

	\vspace{-8pt}
	\caption{Throughput of high-vertical-throughput router architecture.}
	\label{fig:throughputPseudo}
\end{figure}

Using the novel high vertical-throughput router architecture, the throughput of packets can be increased if the slower layer is contained in the path. In fact, the throughput will be as high as in the faster layer, if area for links and routers is expendable. This is shown in Fig.~\ref{fig:throughputPseudo}. For a transition from a slower to a faster layer (shown on left-hand side), the packet throughput is not determined by the slower clock frequency because the packet can be completely transmitted once it is available. For the opposite direction (right-hand side), the throughput is also not determined by the slower clock, since the complete packet  becomes available at the faster router.

\subsection{Area and power of proposed router architecture and routing algorithms}\label{sec:results:area}

We synthesize the baseline router using \XYZ\ routing and the proposed \emph{high vertical-throughput router} using \AlgorithmI{/}\AlgorithmII\ routing in a commercial \unit[45]{nm} ULV mixed-signal technology (We only synthesize for mixed-signal since the routers in the digital faster layer do not have a modified crossbar). The same crossbar optimizations are applied for both conventional and vertical-high throughput architectures. We assume a $4\!\times\!3\!\times\!3$ NoC with one digital layer. The flit width is \unit[32]{b}, the input buffer depth is eight, the flow-control is credit-based and four virtual channels are only used in the digital layer. Both architectures, the proposed high vertical-throughput router as well as the baseline baseline, can run with a maximum frequency of \unit[500]{MHz}. Area and power results are shown in Tab.~\ref{tab:areapower} and elaborate as follows:

The \emph{area overhead} of the proposed routing algorithms is negligible. In fact \AlgorithmI\ routing has \unit[-1.32]{\%} overhead compared to \XYZ\ routing. For \AlgorithmII,  the area is only increased by three gate equivalents, which affects the whole router area by less than \unit[-2.38]{\%}. The area of the crossbar and the input buffers depends on the clock frequency of the digital routers. To bridge to a clock frequency of \unit[1]{GHz} in the faster layer ($c_f{=}2$), the total area required for the routers is increased by \unit[2.1]{\%}. If the routers in the fast layer are clocked at \unit[2]{GHz} ($c_f{=}4$), the total area increases by \unit[10.6]{\%}.

\emph{Dynamic power savings} are possible. We simulated the aforementioned NoC with 1M clock cycles, injecting uniform random traffic at \unit[4]{\%} injection rate. The digital layers is implemented in \unit[15]{nm} digital technology and the mixed signal layer in \unit[45]{nm} ULV node. For a clock difference of $c_f = 2$, the proposed routing algorithms saved \unit[41.1]{\%} dynamic power in comparison to conventional XYZ-routing; For a clock difference of $c_f = 4$, the proposed routing algorithms saved \unit[30.3]{\%} dynamic power.

\subsection{Case Study}\label{sec:results:casestudy}

We analyze our approach for a 3D VSoC based on \cite{Zarandy.2011} with four layers as shown in Figure~\ref{fig:casestudy}: The first layer is a sensing die, implementing a \unit[180]{nm} CIS (CMOS Imaging Sensor). The second layer implements nine analog digital converters (ADCs) and three analog accelerators \cite{Jia.2017} in \unit[90]{nm} mixed-signal node. The third layer implements 6 processors and 6 SIMD (single instruction multiple data) acceleration units in \unit[15]{nm} digital node. In the fourth layer there are 12 processor cores in \unit[30]{nm} digital node. The first and second layer are connected via point-to-point links. The second, third and fourth layer are connected via a 3D NoC with \unit[32]{b} wide links, 8 flit deep buffers and 4 VCs. Packets are 32 flits long with one flit header. Routers in the digital layer are clocked at \unit[1]{GHz} and in the mixed-signal layer at \unit[0.5]{GHz}.

The 3D VSoC implements an image processing pipeline for face recognition. The image sensor records at 720p. The ADCs send the digital raw image to the processors in the third layer, which apply Bayer filter. Then, the SIMD units reduce the resolution by a factor of 4 to increase feature extraction speed. The result is transmitted to the analog accelerators in the second layer, which extract features using Viola-Jones algorithm \cite{Viola-2001:1}. The resulting region of interest is transmitted to the fourth layer, in which the processors execute Shi and Tomasi algorithm \cite{Shi-1994:1} to find features to track and Kande-Lucas-Tomasi algorithm \cite{Lucas-1991:1} tracks them. Work is split up equally among the available resources in each step.

We simulate the VSoC's NoC using the described application traffic. Thereby, we compare \AlgorithmI\ and \AlgorithmII\ with \XYZ\ routing. We simulate 3M clock cycles in the digital layers and 1.5M in the mixed-signal layer. We measure the average flit latency as \unit[145.91]{ns}  for conventional routing and as \unit[64.46]{ns} for the proposed routing. This equates to a speedup of $2.26\times$. Using the models, we calculate a theoretical speedup of $2.28\times$ under zero load. Average delay for whole packets is reduced from \unit[229.23]{ns} to \unit[123.07]{ns}, which is a speedup of $1.86\times$.

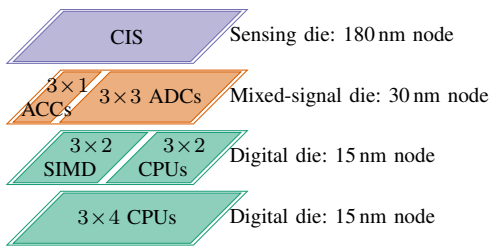
\begin{figure}
	\centering
%
%
%
	
\begin{tikzpicture}[scale = .62,
axis/.style={-latex,black, thick}, digital/.style={thick, black}, digitalLines/.style = {black}]

\tikzmath{\l1 = 0;};
\tikzmath{\l2 = 1.3;};
\tikzmath{\l3 = 2.6;};
\tikzmath{\l4 = 3.9;};

\draw [col1] (0,\l1, 0) -- (4,\l1,0) --(4,\l1,3) -- (0,\l1,3) -- cycle;
\node [desc, anchor = west] (source) at (4,\l1,1.5) {Digital die: \unit[15]{nm} node};
\draw [col1, fill = col1!50] (0.1, \l1, 0.1) -- (3.9,\l1,0.1) --(3.9,\l1,2.9) -- (0.1,\l1,2.9) -- cycle;
\node [desc, anchor = center] (source) at (2, \l1, 1.5) {$3\!\times\!4$ CPUs};

\draw [col1] (0,\l2, 0) -- (4,\l2,0) --(4,\l2,3) -- (0,\l2,3) -- cycle;
\node [desc, anchor = west] (source) at (4,\l2,1.5) {Digital die: \unit[15]{nm} node};
\draw [col1, fill = col1!50] (0.1, \l2, 0.1) -- (1.9,\l2,0.1) --(1.9,\l2,2.9) -- (0.1,\l2,2.9) -- cycle;
\node [desc, anchor = center, align = center] (source) at (1, \l2, 1.5) {\quad$3\!\times\!2$\\\hspace{-1em}SIMD};
\draw [col1, fill = col1!50] (2.1, \l2, 0.1) -- (3.9,\l2,0.1) --(3.9,\l2,2.9) -- (2.1,\l2,2.9) -- cycle;
\node [desc, anchor = center, align = center] (source) at (3, \l2, 1.5) {\hspace{1em}$3\!\times\!2$\\\hspace{-1em}CPUs};

\draw [col2] (0,\l3, 0) -- (4,\l3,0) --(4,\l3,3) -- (0,\l3,3) -- cycle;
\node [desc, anchor = west] (source) at (4,\l3,1.5) {Mixed-signal die: \unit[30]{nm} node};
\draw [col2, fill = col2!50] (0.1, \l3, 0.1) -- (.9,\l3,0.1) --(.9,\l3,2.9) -- (0.1,\l3,2.9) -- cycle;
\draw [col2, fill = col2!50] (1.1, \l3, 0.1) -- (3.9,\l3,0.1) --(3.9,\l3,2.9) -- (1.1,\l3,2.9) -- cycle;
\node [desc, anchor = center, align = center] (source) at (.5, \l3, 1.5) {\hspace{1em}$3\!\times\!1$\\\hspace{-1em}ACCs};
\node [desc, anchor = center] (source) at (2.5, \l3, 1.5) {$3\!\times\!3$ ADCs};

\draw [col3] (0,\l4, 0) -- (4,\l4,0) --(4,\l4,3) -- (0,\l4,3) -- cycle;
\node [desc, anchor = west] (source) at (4,\l4,1.5) {Sensing die: \unit[180]{nm} node};
\draw [col3, fill = col3!50] (0.1, \l4, 0.1) -- (3.9,\l4,0.1) --(3.9,\l4,2.9) -- (0.1,\l4,2.9) -- cycle;
\node [desc, anchor = center] (source) at (2, \l4, 1.5) {CIS};

\end{tikzpicture}

	\caption{3D VSoC case study based on \cite{Zarandy.2011}.} 
	\label{fig:casestudy}
\end{figure}
\section{Discussion}\label{sec:discussion}

First, we discuss the model accuracy as those are the basis for the subsequent evaluation of the routing algorithms. The aim of the models is to estimate the impact of heterogeneity on NoCs. Figs. \ref{fig:nocrouterScalingArea} and \ref{fig:nocrouterScalingTiming} demonstrate a very good fit of the models for available nodes to academia. The area model has small RSMEs, which is a result of the model's physical foundation. It was not beneficial to add a linear term to this model; this increases the RMSEs. The timing model is empirical and thus the fit is overall less accurate than the area model fit, shown by higher RMSEs. The model converges to the target maximum clock frequency, as desired. If more modern technology nodes were available, either a better model with a physical foundation could be found or the fit of our model could be improved. Nonetheless, the model serves its purpose here: both the timing and the area model provide sufficient accuracy to assess the influence of heterogeneous integration on routing, as we further quantify. Therefore, we apply the fitted data to calculate the propagation speed $\omega$ for a predictive technology. This is shown in Fig.~\ref{fig:propagationSpeed}. Comparing predictive technology calculated with the models to the synthesis results for \unit[130]{nm} commercial mixed-signal and \unit[90]{nm} -- \unit[28]{nm} commercial digital technologies yields an accuracy of between $\unit[1.4]{\%}$ and $\unit[7.8]{\%}$. This supports that the proposed models are valid. We also propose models for latency and throughput. That they are accurate is given by construction and validated using simulations. The results are shown in Figs.~\ref{fig:SpeedupAlgo1} and \ref{fig:SpeedupAlgo2}. The results for latency will be identical, regardless if obtained from simulations or from the proposed model. Therefore, the communication models are precise under zero load. There is no need to model the behavior under load for the purpose of this paper. Of course, the models will not be valid if further traffic is injected and the assumption of zero load is violated. However, our model can also be extended to cover dynamic effects by applying a queueing model \cite{Kiasari.2013}. This is not required here because the unique effects of heterogeneity have already been revealed under zero load. Load is applied in our case study and our routing show a latency enhancement, as well. In fact, we see a speedup of $2.26\times$ in simulations under load, while our models predict a speed-up of $2.28\times$. This shows that our models are accurate enough to find useful routing algorithms under real conditions, even though they only account for zero load within our case study. Thus, by means of our model, we are able to conduct powerful routing strategies and architectures for heterogeneous 3D interconnect. 

Second, the exemplary implementations of routing algorithms and router architectures are evaluated. The aim of the implementations is to mitigate the negative effects of heterogeneity (worse latency and throughput), with as few area costs as possible. The largest limitations of heterogeneity emerge if the difference between mixed-signal and purely digital technology are large; therefore we focus on a chip using \unit[130]{nm} commercial mixed-signal technology and \unit[28]{nm} commercial digital technology. The results can also be applied to any other combination of technologies with similar relative technology scaling factor $\Xi$. The proposed routing algorithms \AlgorithmI\ and \AlgorithmII\ provide up to $6.5\times$ latency reductions for packets from routers in the mixed-signal nodes to routers in the digital layer and up to $1.79\times$ latency reductions for packets within the layer in the mixed-signal node in comparison to dimension order routing. This is shown in Figs.~\ref{fig:SpeedupAlgo1} and \ref{fig:SpeedupAlgo2}. For \AlgorithmII, there is a performance penalty for distances below $\Phi$ (Eq.~\ref{eq:Phi}) of up to $45\%$, as expected (see Fig.~\ref{fig:SpeedupAlgo2}, left-hand side). The threshold distance shrinks for more advanced technology nodes, which is also expected. The \XYZ\ outperforms \AlgorithmII\ for low technology differences for all distances. 

We compare the router for a practical scenario with a clock difference of 2 between layers in different nodes to show advantages of our approach. The results are summarized in Tab.~\ref{tab:areapower}. For a real-world based benchmark, we simulate a face recognition image processing pipeline on a 3D VSoC based on \cite{Zarandy.2011} with \unit[45]{nm} mixed-signal technology and \unit[15]{nm} digital technology. The proposed vertical high-throughput router offers $2.26\times$ better latency and an increased throughput of up to $2\times$, in simulations, at \unit[2.1]{\%} area increase comparing to a standard router for \XYZ\ routing. If a larger throughput increase is desired, additional area costs must be expended. While the area is increased, dynamic power is saved: We showed \unit[41.4]{\%} dynamic power in simulations. The performance speedups and power savings demonstrates the impressive benefit of the proposed approach for typical applications of heterogeneous 3D SoCs.

{\footnotesize \sffamily
	\begin{table}
		\caption{PPA comparison of proposed routing algorithms and high-throughput routers to conventional router.}
		\label{tab:areapower}
		\centering
		\begin{tabular}{|r|r|r|r|}
			\hline
			\multicolumn{2}{|c|}{\textbf{Performance}} & \textbf{Area} & \textbf{Power} \\
			\textit{throughput}& average \textit{latency}& total \textit{area} & dynamic \textit{power}\\
			increase &  speedup of flits & increase & savings \\
			\hline
			\textcolor{col1}{\textbf{2$\times$}} & 	\textcolor{col1}{\bfseries 2.26$\times$}&	\textcolor{red!90}{\bfseries \unit[2.1]{\%}}&	\textcolor{col1}{\bfseries \unit[41.4]{\%}}\\
			\hline
		\end{tabular}
	\end{table}
}

To summarize, \AlgorithmI\ and \AlgorithmII, in combination with the novel router architectures, have small area overhead and better performance than state-of-the-art both in theoretical and practical evaluations. Therefore, limitations of heterogeneity on routing in 3D NoCs are mitigated. Only by an integrated design of routing strategies and architectures, we are able to design an efficient and powerful heterogeneous 3D interconnect.

\section{Conclusion}

Heterogeneous 3D SoCs need to combine disparate technologies, e.g.\ mixed-signal and purely digital technologies; However, the impact of heterogeneity on interconnection networks was previously not considered. We show that varying throughput and latency of NoCs in layers in disparate technologies drastically degrades network performance. To prove this, models for area and timing of routers, and for latency and throughput under zero load have been proposed. The models are well-founded and express the relevant effects of heterogeneity on routing; the model accuracy is high and shows an error of \unit[1.4]{\%}-\unit[7.8]{\%} for an exemplary technology scenario. Based on the model's findings, we develop principles for routing in heterogeneous 3D SoCs. We show their practical applicability by proposing two new exemplary routing algorithms. These reduce the network latency for packets between nodes in mixed-signal and purely digital technologies and between nodes in a mixed-signal layer by utilization of faster transmission speeds in digital layers. For an exemplary SoC, with layers in commercial \unit[28]{nm} digital and commercial \unit[130]{nm} mixed-signal technology, we achieve a latency reduction of up to $6.5\times$ at negligible area overhead in comparison to conventional dimension ordered routing. We further propose a novel vertical high-throughput router architecture and a vertical link design to overcome the throughput limitations, which increase throughput by up to $2\times$ at $6\%$ reduced router area costs for the same exemplary set of technologies. Within simulation of a case study for a 3D VSoC using \unit[30]{nm} mixed-signal and \unit[15]{nm} digital technologies implementing a face recognition algorithm, we could validate our theoretical findings with a speedup of $1.86\times$ to $2.26\times$ for average latency and $2\times$ for throughput. We also showed \unit[41.4]{\%} reduced dynamic power in simulations using uniform random traffic. Summing up, the proposed co-design of routing algorithms and router architectures mitigate limitations of NoCs in heterogeneous 3D SoC. It allows much better performance and dynamic power consumption at small to negligible area overhead.

\section*{Acknoledgements}
This work is funded by the German Research Foundation (DFG) projects PI 447/8 and GA 763/7.

\bibliographystyle{IEEEtran}
\bibliography{bibliography_short}

\begin{IEEEbiography}[{\includegraphics[width=1in,height=1.25in,clip,keepaspectratio]{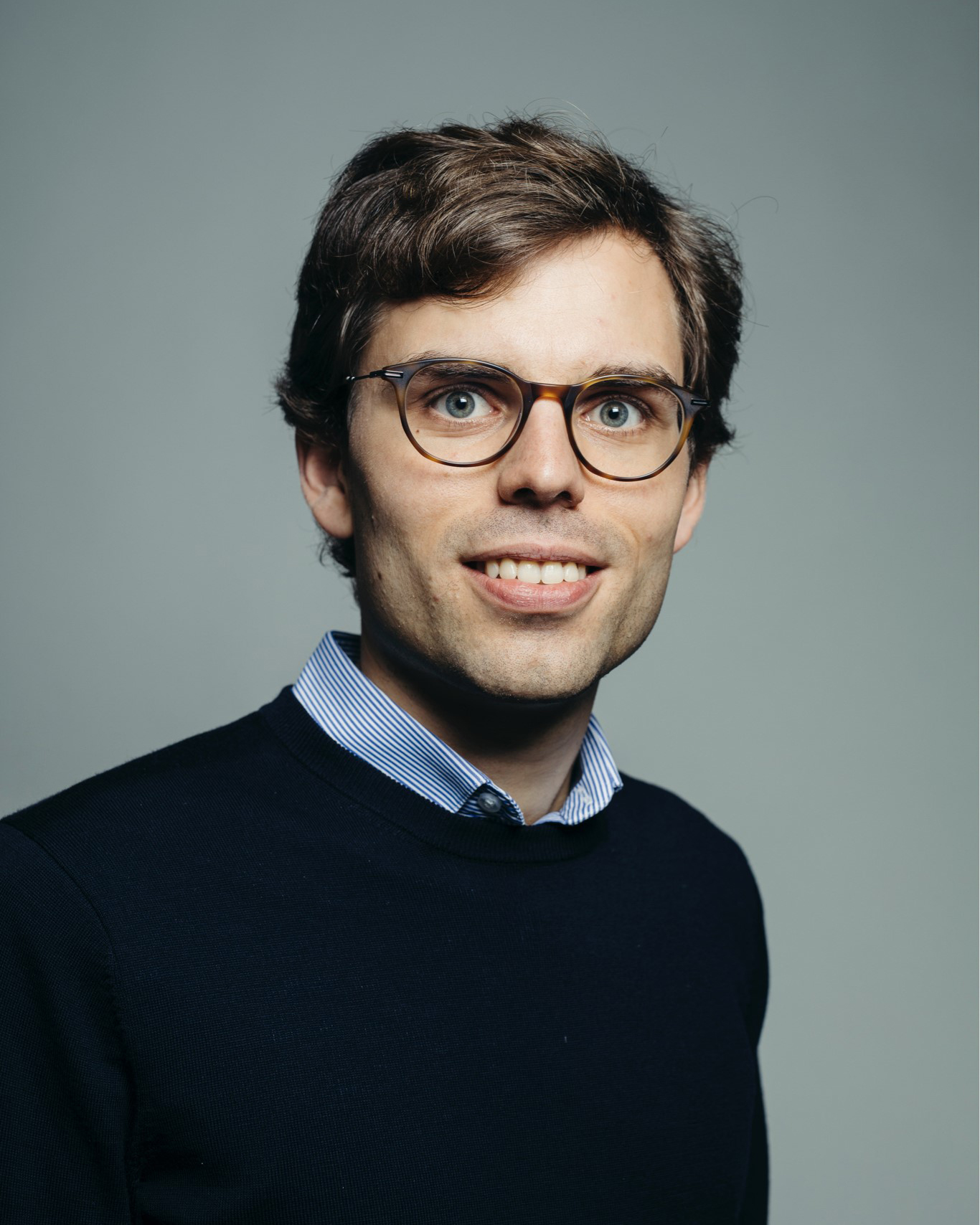}}]{Jan Moritz Joseph} received the B.Sc. degree in Medical Engineering in 2011 and the M.Sc. degree in Informatics in 2014 from the Universität zu Lübeck, Germany. From 2008 to 2014, he was a scholarship holder of The German National Merit Foundation. He is currently as a research assistant at the Otto-von-Guericke-Universität Magdeburg, Germany. His focus is on 3D integration. Currently, he researches heterogeneous integration, interconnects and NoCs.
\end{IEEEbiography}

\begin{IEEEbiography}[{\includegraphics[width=1in,height=1.25in,clip,keepaspectratio]{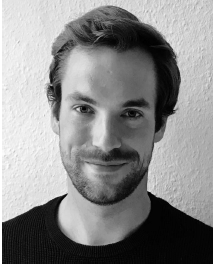}}]{Lennart Bamberg} Lennart Bamberg received the B.Sc. and M.Sc. degree in Electrical and Information Engineering from the University of Bremen, Germany, in 2014 and 2016, respectively. He is currently 	working towards the Ph.D. degree at the University of Bremen, Germany, where he is employed since 2016 as a teaching and research associate. In 2019, Lennart Bamberg  joined the Georgia Institute of Technology, Atlanta (USA) for four months as a visiting scholar. Lennart Bamberg received the Best Paper Award at PATMOS 2017 and PATMOS 2018. His research interests include low-power design, communication-centric design and heterogeneous 3D SoCs.
\end{IEEEbiography}

\begin{IEEEbiography}[{\includegraphics[width=1in,height=1.25in,clip,keepaspectratio]{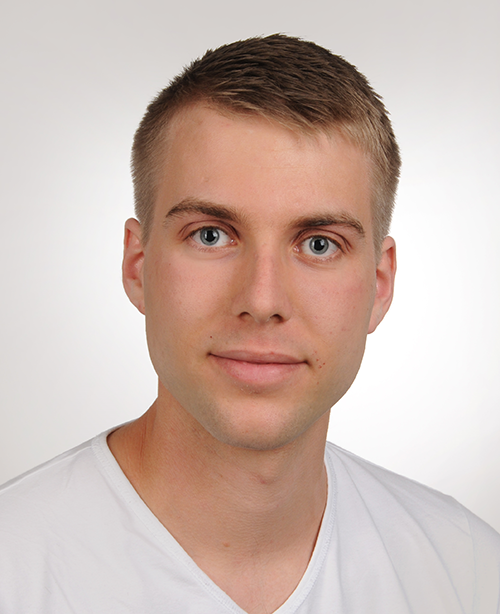}}]{Dominik Ermel}
received the B.Sc. degree in Mathematics in 2016 from the Otto-von-Guericke-University Magdeburg. He is currently working on his M.Sc. degree. He is interested in mathematical optimization and has been applying combinatorial optimization on the topics heterogeneous 3D integration and Networks-on-Chip.
\end{IEEEbiography}

\begin{IEEEbiography}[{\includegraphics[width=1in,height=1.25in,clip,keepaspectratio]{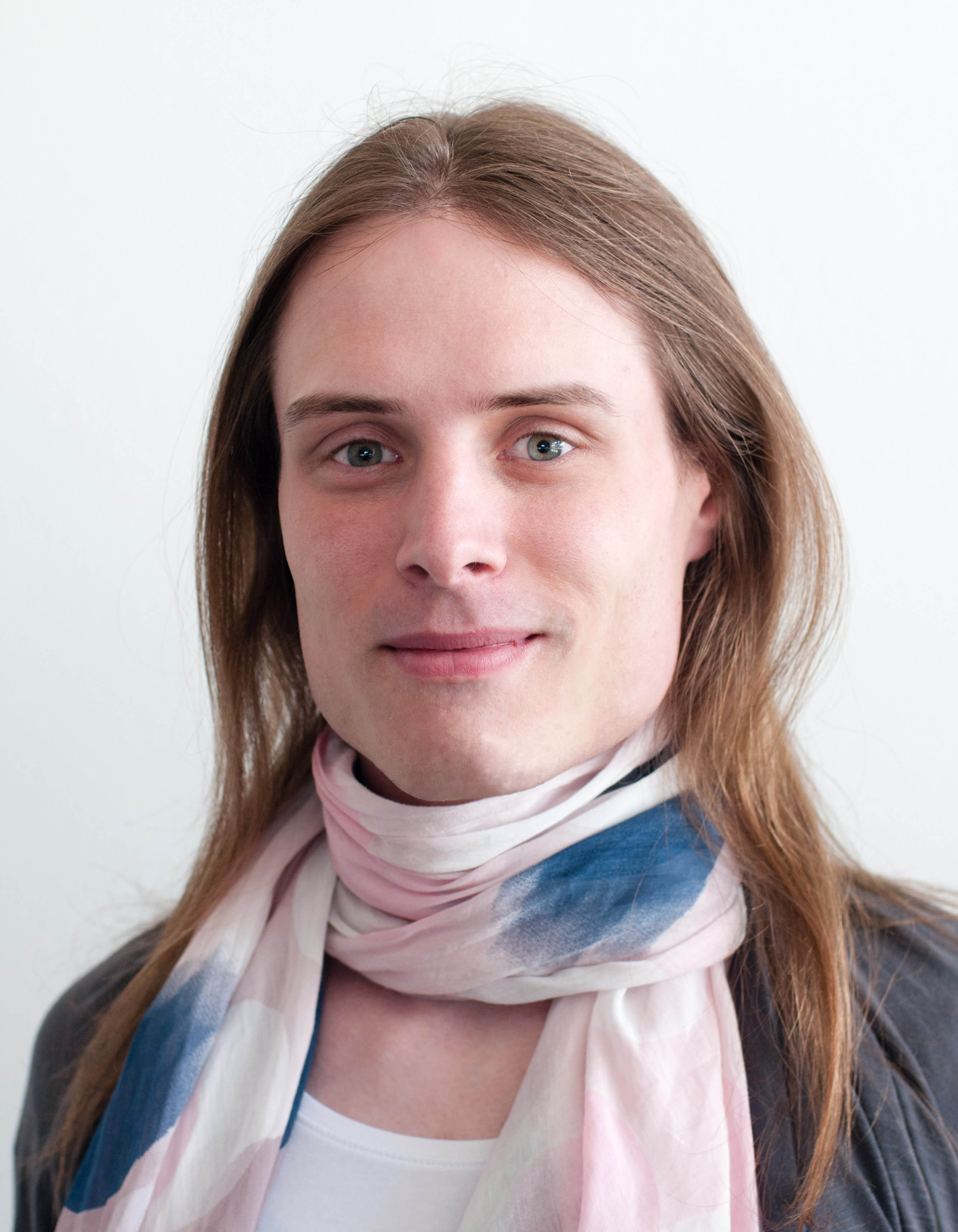}}]{Anna Drewes}
received the B.Sc. and M.Sc. degrees in computer science form the University of Lübeck, Germany, in 2015 and 2017, respectively.
She is currently pursuing a Ph.D. while working as a research assistant at the Institute for Information Technology and Communications at the Otto-von-Guericke-University Magdeburg, Germany.
Her research interests include communications infrastructure and interconnects, especially for FPGAs, as well as the use of heterogeneous systems for database query processing.
\end{IEEEbiography}

\begin{IEEEbiography}[{\includegraphics[width=1in,height=1.25in,clip,keepaspectratio]{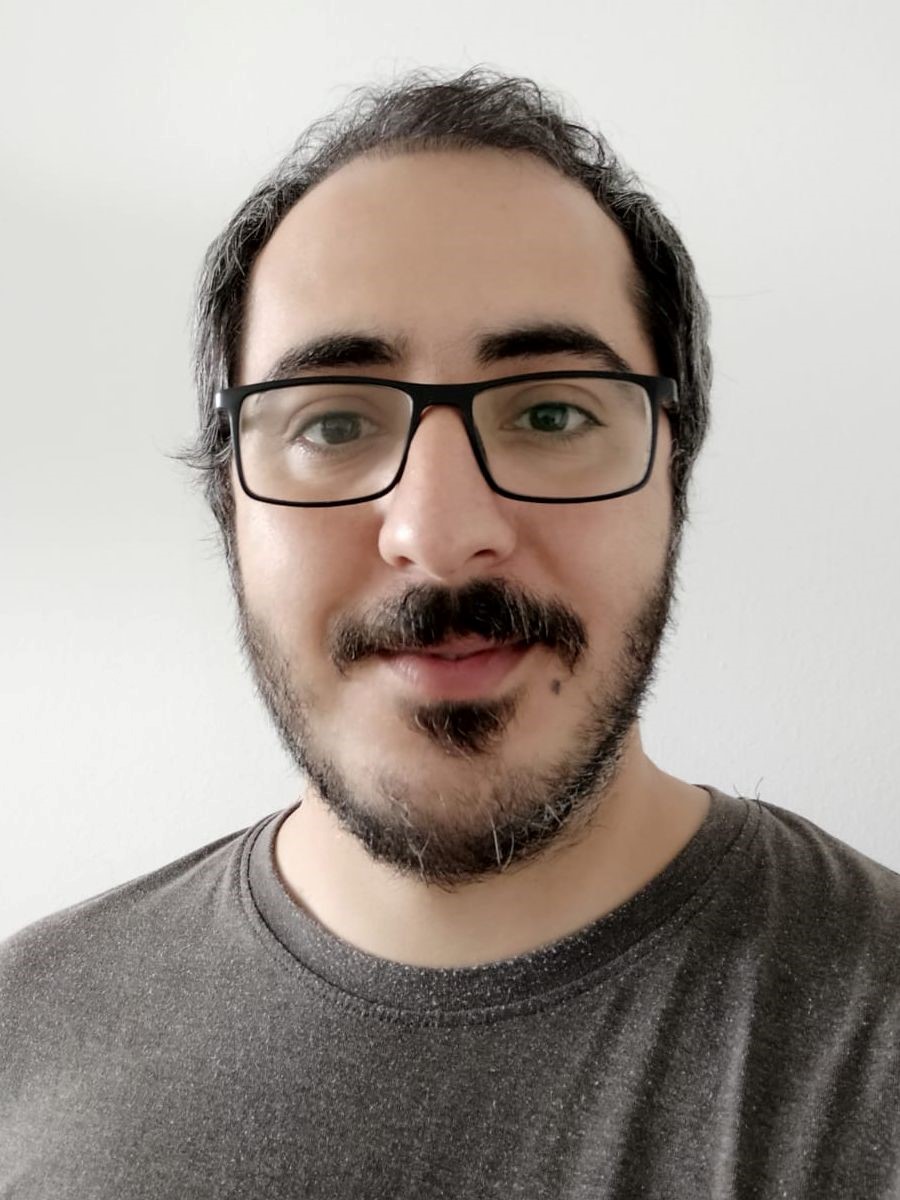}}]{Behnam Razi Perjikolaei}
received his B.Sc. and M.Sc. in computer engineering and computer systems architecture from Shahid Bahonar University of Kerman, Iran in 2007 and from IAU Science and Research branch Tehran, Iran, in 2012, respectively. From 2012 to 2016 he worked in the industrial automation department of ACECR Sharif University branch, Iran. Currently he is pursuing his second M.Sc. degree in Control, Microelectronics and Microsystems in University of Bremen, Germany. His current research interests include Network-on-Chip communication architectures, especially for FPGA and heterogeneous 3D architecture. 
\end{IEEEbiography}
\begin{IEEEbiography}[{\includegraphics[width=1in,height=1.25in,clip,keepaspectratio]{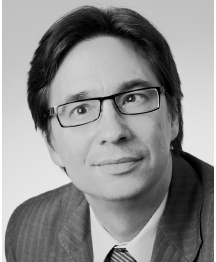}}]{Prof. Alberto Garc\'ia-Ortiz} obtained  the diploma degree in Telecommunication Systems from Universitat Politecnica de Valencia in 1998.  After working for two years at Newlogic in Austria, he started the Ph.D. at the Institute of Microelectronic Systems, Technische Universität Darmstadt, Germany.  In 2003, he received  the  Ph.D.  degree  with  summa cum laude. From 2003 to 2005, he worked as a Senior Hardware Design Engineer at IBM Deutschland Development and Research in Böblingen.  After that he joined AnaFocus in Seville, Spain.  Since 2011, he is full professor for the chair of integrated digital systems at the University of Bremen. Dr. Garcia-Ortiz  received  the  Outstanding dissertation award in 2004 from the European Design and Automation Association.  In 2005, he received from IBM an innovation award for contributions to leakage estimation. He  serves as editor and reviewer of several  conferences, journals, and projects. His interests include low-power design, communication-centric design, SoC integration, and variations-aware design.
\end{IEEEbiography}
\begin{IEEEbiography}[{\includegraphics[width=1in,height=1.25in,clip,keepaspectratio]{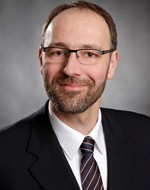}}]{Prof. Thilo Pionteck}
is holding the chair for hardware-oriented computer science at the Otto-von-Guericke-Universität Magdeburg, Germany. He received his Diploma degree in 1999 and his Ph.D. (Dr.-Ing.) degree in Electrical Engineering from the Technische Universität Darmstadt, Germany. In 2008, he was appointed as assistant professor for Integrated Circuits and Systems at the Universität zu Lübeck. From 2012 to 2014, he was substitute of the Chair of Embedded Systems at the Technische Universität Dresden and of the Chair of Computer Engineering at the Technische Universität at Hamburg-Harburg, Germany. In 2015 he was appointed as professor of the Chair of Organic Computing at the Universität zu Lübeck, Germany, with research focus on adaptive digital systems. He was appointed to the Otto-von-Guericke Universität Magdeburg, Germany in 2016. His research work focuses on Network-on-Chips, adaptive system design, runtime reconfiguration, and hardware/software co-design.
\end{IEEEbiography}
%


\end{document}